\PassOptionsToPackage{prologue,dvipsnames}{xcolor}
\documentclass[natbib,acmsmall,nonacm,11pt]{acmart}
\pdfoutput=1 
\usepackage[T1]{fontenc}
\usepackage[utf8]{inputenc}
\usepackage[english]{babel}
\usepackage{geometry}
\def\RATIO{50}
\def\ROUND#1{O_{#1}(1)}

\usepackage{graphicx,framed,calc,amsmath,amsthm,algorithm,algpseudocode}
\usepackage{wrapfig,etoolbox,anyfontsize,bm,textcase}
\usepackage{balance}
\usepackage[dvipsnames]{xcolor}
\usepackage{float,csquotes,xspace}

\usepackage{multicol,multirow,varwidth,outlines}

\usepackage{tikz,tikz-3dplot,tikz-cd,tkz-euclide,pgf,pgfplots,adjustbox}
\usetikzlibrary{math,patterns,positioning,arrows,arrows.meta,graphs,graphs.standard,decorations.pathmorphing,tikzmark,calc}
\pgfplotsset{compat=newest}


\NewDocumentCommand\set{sm}{\IfBooleanTF#1{\{{#2}\}}{\left\{{#2}\right\}}}
\NewDocumentCommand\ceil{sm}{\IfBooleanTF#1{\lceil{#2}\rceil}{\left\lceil{#2}\right\rceil}}
\NewDocumentCommand\floor{sm}{\IfBooleanTF#1{\lfloor{#2}\rfloor}{\left\lfloor{#2}\right\rfloor}}
\NewDocumentCommand\pare{sm}{\IfBooleanTF#1{({#2})}{\left({#2}\right)}}
\NewDocumentCommand\range{smm}{\IfBooleanTF#1{\set*{{#2},\dots,{#3}}}{\set{{#2},\dots,{#3}}}}
\def\pw{\mathrm{pw}}

\def\poly{\mathrm{poly}}

\def\myfootnotemark#1#2{%
  {\let\thefootnote\relax\footnotemark\addtocounter{footnote}{-1}\hspace{-.5ex}}%
  \textsuperscript{\hyperref[#1:#2]{\the\numexpr\thefootnote+#2}}%
}

\def\myfootnotetext#1#2#3{%
  \addtocounter{footnote}{1}%
  \footnotetext{\label{#1:#2}#3}%
}

\definecolor{primary}{HTML}{13293D}
\colorlet{secondary}{Dandelion}
\definecolor{tertiary}{HTML}{fe4e63}
\definecolor{ourresults}{RGB}{200,0,0}

\usepackage[nameinlink,capitalize]{cleveref}
\usepackage[norefs,nocites]{refcheck}
\makeatletter
\newcommand{\refcheckize}[1]{%
  \expandafter\let\csname @@\string#1\endcsname#1%
  \expandafter\DeclareRobustCommand\csname relax\string#1\endcsname[1]{%
    \csname @@\string#1\endcsname{##1}\wrtusdrf{##1}}%
  \expandafter\let\expandafter#1\csname relax\string#1\endcsname
}
\makeatother

\refcheckize{\cref}
\refcheckize{\Cref}

\newcommand{\eps}{\varepsilon}

\newcommand{\A}{\mathcal{A}}
\newcommand{\B}{\mathcal{B}}
\newcommand{\cC}{\mathcal{C}}
\newcommand{\D}{\mathcal{D}}

\newcommand{\N}{\mathbb{N}}

\newcommand{\R}{\mathcal{R}}

\let\oldO\O
\renewcommand{\O}{\mathcal{O}}

\newtheorem{definition}{Definition}[section]
\newtheorem{proposition}[definition]{Proposition}

\newtheorem{theorem}[definition]{Theorem}
\newtheorem{lemma}[definition]{Lemma}
\newtheorem{corollary}[definition]{Corollary}

\newtheorem{claim}[definition]{Claim}
\crefname{claim}{claim}{claims}
\crefname{resta}{claim}{claims}
\crefname{subsection}{subsection}{subsections}

\renewenvironment{proof}{\noindent\textbf{Proof.~}}{{}\hfill$\Box$\\}

\usepackage{thmtools}

\newcommand{\LOCAL}{\textsf{LOCAL}\xspace}
\newcommand{\CONGEST}{\textsf{CONGEST}\xspace}

\DeclareMathOperator{\MDS}{MDS}

\DeclareMathOperator{\MVC}{MVC}

\DeclareMathOperator{\asdim}{asdim}

\let\le\leqslant
\let\ge\geqslant
\let\leq\leqslant
\let\geq\geqslant

\tikzset{%
    draw=primary,
    text=primary,
    vertex/.style={circle,fill=primary,line width=2pt},
}    

\title{Local Constant Approximation for Dominating Set on Graphs Excluding Large Minors}

\author{Marthe Bonamy}
\orcid{0000-0001-7905-8018}
\affiliation{%
  \institution{LaBRI, University of Bordeaux, CNRS}
  \city{Bordeaux}
  \country{France}
}
\email{marthe.bonamy@u-bordeaux.fr}

\author{Cyril Gavoille}
\orcid{0000-0003-3671-8607}
\affiliation{%
  \institution{LaBRI, University of Bordeaux, CNRS}
  \city{Bordeaux}
  \country{France}
}
\email{gavoille@labri.fr}

\author{Timothé Picavet}
\orcid{0000-0002-7129-0127}
\affiliation{%
  \institution{LaBRI, University of Bordeaux, CNRS}
  \city{Bordeaux}
  \country{France}
}
\email{timothe.picavet@u-bordeaux.fr}

\author{Alexandra Wesolek}
\orcid{0000-0003-4841-5937}
\affiliation{%
  \institution{Institut für Mathematik, Technische Universit{\"a}t Berlin}
  \city{Berlin}
  \country{Germany}
}
\email{wesolek@tu-berlin.de}

\date{}
\authorsaddresses{}

\def\minDS{\textsc{Minimum Dominating Set}\xspace}
\def\minVC{\textsc{Minimum Vertex Cover}\xspace}

\begin{abstract}
    We show that graphs excluding $K_{2,t}$ as a minor admit a $f(t)$-round $\RATIO$-approximation deterministic distributed algorithm for \minDS. The result extends to \minVC. Though fast and approximate distributed algorithms for such problems were already known for $H$-minor-free graphs, all of them have an approximation ratio depending on the size of $H$. To the best of our knowledge, this is the first example of a large non-trivial excluded minor leading to fast and constant-approximation distributed algorithms, where the ratio is independent of the size of $H$. 
    A new key ingredient in the analysis of these distributed algorithms is the use of \textit{asymptotic dimension}.
\end{abstract}

\begin{document}


\maketitle

\tableofcontents

\newpage

\section{Introduction}

 \minDS (MDS) (and its weaker version, \minVC (MVC)) is a famous minimization problem on graphs, known to be NP-complete even in cubic planar graphs~\cite{GJ79,KYK80}. The goal is to find a smallest subset of vertices that intersects all radius-1 balls (MDS) or all edges (MVC). 
 
Applications of vertex covers, dominating sets, and other types of covers can be found in the context of wireless sensor networks \cite{CIQC02,Krishnamachari05,AFK21}. There, the goal is to minimize energy by keeping as few devices active as possible while maintaining the ability to awake sleeping devices through an active neighbor. For this purpose, the distributed version is also important. 

\paragraph{The \LOCAL Model.}

In this paper, we consider distributed algorithms in the \LOCAL model, popularized by Linial in his seminal papers~\cite{Linial87,Linial92}. In this model, the network is represented by an undirected connected graph $G$, the edges representing reliable communication links between computing devices (the vertices) that work in synchronous rounds. At each round, a vertex can exchange messages with each of its neighbors and perform arbitrary computations based on the information it has. Messages have no size limit, in contrast to the \CONGEST model. At the start of the algorithm, the processors each have a copy of the algorithm and a $O(\log{n})$-bit identifier, where $n$ is the number of vertices in the graph $G$. The main complexity measure in the \LOCAL model is the number of rounds to achieve a given task, taken as the maximum over all vertices. This measure gives an indication on the local nature of a problem, as it captures the minimum value $r$ such that each vertex can reach a good decision based on its radius-$r$ neighborhood. 

\paragraph{Fast Algorithms.}

In any $n$-vertex graph $G$, it is possible to $(1+\eps)$-approximate a MDS for $G$ in $\poly(\eps^{-1}\log{n})$ rounds by combining the techniques of \cite{GKM17} and of \cite{RG20}, see \cite[Cor.~3.11]{RG20}.
For more specific graphs, $O(\log^*{n})$ rounds may suffice. This is for instance the case in planar graphs~\cite{CHW08a}, or more generally in $K_t$-minor-free graphs~\cite{CHW18} and in sub-logarithmic expansion graphs~\cite{ASS19} -- we emphasize that hidden constants in the big-$O$ notation for the number of rounds    depend on $\eps$ and $t$. Conversely, \cite{CHW08a} showed that $o(\log^*{n})$ rounds do not suffice for computing an $(1+\eps)$-approximation of MDS on a cycle in the  \LOCAL model. More generally, using a different technique inspired by Linial, \cite{LW08} showed that the approximation-ratio times the round-complexity must be $\Omega(\log^*{n})$ for any approximation \LOCAL algorithm for MDS in unit-disk graphs. 

\paragraph{Constant-Round Algorithms.}

Because of the lower bound of \cite{KMW16} in general graphs, achieving constant ratio approximation in a constant number of rounds is not possible. More precisely, every constant-approximation \LOCAL algorithm requires $\Omega(\sqrt{\log{n}/\log{\log{n}}})$ rounds, and this holds for MDS and MVC. Therefore, we need to focus on restricted graph classes to obtain such results. 

The literature is abundant in this direction. For instance, on regular graphs, i.e., graphs where all vertices have the same degree, we achieve a $2$-approximation for MVC in~$0$ rounds (take all vertices\footnote{This holds by observing that such a graph contains $kn/2$ edges where $k$ is the degree of each vertex, while a set on $p$ vertices intersects at most $pk$ edges.}). Similarly, a $6$-approximation in unit-disk graphs can be achieved by taking all the vertices incident to an edge -- See the excellent survey of \cite{Suomela13} and the references therein. For the more difficult MDS problem, distributed algorithms have been developed for various classes of graphs, including (but not limited to): outerplanar graphs~\cite{BCGW21}, planar graphs~\cite{LPW13,HLS14,HKOSV25}, bounded-genus graphs~\cite{ASS19}, graphs excluding topological minors~\cite{CHW18}, graphs with sublogarithmic expansion~\cite{ASS19} or with bounded expansion~\cite{KSV21,HKOSV25}. For instance, the approximation ratio for MDS in planar graphs has been improved from $130$~\cite{LPW13} to $52$~\cite{Wawrzyniak14}, and recently down to $11+\eps$~\cite{HKOSV25}.

\paragraph{$H$-minor-free With Large $H$.}

Most of the graph classes cited above can be expressed as $H$-minor-free graphs for some specific minor $H$ (cf. \Cref{table:results}), but the results for graphs with bounded expansion are more general.
More precisely, \cite{KSV21} presented a constant-round \LOCAL algorithm with approximation ratio ${\nabla_1(G)}^{\O(t\nabla_1(G))}$ if $G$ excludes $K_{t,t}$ as subgraph, where $\nabla_r(G)$ is the maximal edge density of a depth-$r$ minor of $G$ and $t = \O(\nabla_1(G))$.
\cite{HKOSV25} have improved the approximation ratio to $\nabla_0(G)\cdot{\nabla_1(G)}^{\O(s\nabla_1(G))}$ if $G$ excludes $K_{s,t}$ as subgraph, at the cost of a larger $\O_t(1)$-round complexity. If $G$ excludes $K_{3,t}$ as subgraph, the approximation ratio improves to $(2+\eps)\cdot (2\nabla_1(G) + 1)$ for every $\eps>0$, where the round complexity is $\O_{\eps,t}(1)$.

Obviously, if $G$ excludes $H$ as minor, it excludes $H$ as a depth-$r$ minor. As a consequence, $\nabla_r(G) \le \delta(H)$, where $\delta(H)$ is the maximum edge density of a graph excluding $H$ as minor. It is well-known~\cite{Kostochka84,Thomason01} that $\delta(K_t) = \Theta(t\sqrt{\log{t}}\,)$, and more generally $\delta(H) = \Theta(t\sqrt{\log{d}}\,)$~\cite{RW16}, where $t = |V(H)|$ and $d = |E(H)|/t < \pw(H)+1$. So, $\delta(K_{s,t}) = \Theta(t\sqrt{\log{s}}\,)$. More specifically, $\delta(K_{3,t}) \le (t+3)/2$~\cite{KP10}.

Therefore, for graphs excluding $K_t$ as minor, the result of \cite{KSV21} implies an approximation ratio of $t^{\O(t^2\sqrt{\log{t}}\,)}$ with $\O(1)$-round complexity. For $K_{s,t}$-minor-free, \cite{HKOSV25} implies an approximation ratio of $t^{\O(st\sqrt{\log{s}}\,)}$ with $\O_t(1)$-round complexity.
For $K_{3,t}$-minor free graphs, the approximation ratio becomes $(2+\eps)\cdot (t+4)$. See \Cref{table:results} for a compilation of best known results for various $H$.

\paragraph{Our Contributions.}

In this paper we concentrate our attention on $H$-minor-free graphs when $H$ has many vertices, that is, $t$ vertices for some arbitrarily large parameter $t\in \mathbb{N}$. To the best of our knowledge (see \Cref{table:results}), no $f(t)$-round and constant-approximation \LOCAL algorithm for MDS in $H$-minor-free graphs is known, excepted perhaps for the trivial case where $H$ is a subgraph of a path with $t$ vertices. Indeed, in this case the graph $G$ has diameter at most $t-1$, and thus a MDS can be solved exactly in $t-1$ rounds\footnote{In the \LOCAL model, after $D$ rounds of communication, each vertex $u$ of a diameter-$D$ graph knows entirely $G$ and its identifier in $G$. After this communication step, $u$ can therefore compute an optimal dominating set in a consistent way with centralized brute-force and deterministic algorithm.}.

For MDS, we show that:

\begin{itemize}
    \item $K_{2,t}$-minor-free graphs have a $\ROUND{t}$-round $\RATIO$-approximation \LOCAL algorithm.
    \item These graphs also have a $O(1)$-round $(2t+1)$-approximation \LOCAL algorithm.    
\end{itemize}

All the algorithms are deterministic, since in general constant-round randomized \LOCAL algorithms are not possible if high probability guarantee is required. 

\begin{table*}[htbp!]
\begin{center}
\renewcommand{\arraystretch}{1.1}
\def\R#1{\textcolor{ourresults}{#1}}%
\begin{tabular}{l|c|c|c}
{\textsc{minor-free graphs}} & \textsc{approx. ratio} & \textsc{\#rounds} & \textsc{references}\\
\hline
\hline
trees ($K_3$) & $3$  & $2$ & Folklore~\myfootnotemark{foot}{1}\\
outerplanar ($K_4$, $K_{2,3}$) & $5$ & $2$ & \cite{BCGW21}\\
planar ($K_5$, $K_{3,3}$) & $11+\eps$  & $O_{\eps}(1)$ & \cite{HKOSV25}\\
\hline
\hline
$K_{1,t}$-minor-free & $t$ &  $0$ & Folklore~\myfootnotemark{foot}{2}\\
\R{$K_{2,t}$-minor-free} & \R{$2t-1$} & \R{$3$} & \R{\textbf{Theorem~\ref{TH:3ROUND}}}\\
\R{$K_{2,t}$-minor-free} & \R{$\RATIO$}  & $\ROUND{t}$ & \R{\textbf{Theorem~\ref{TH:MAIN}}} \\
$K_{s,t}$-minor-free & $t^{O(st\sqrt{\log{s}})}$  & $O_{\eps,t}(1)$ & \cite{HKOSV25}\\
$K_{t}$-minor-free & $t^{O(t^2\sqrt{\log{t}}\,)}$  & $7$ & \cite{KSV21}\\
\hline
\hline
\end{tabular}
\caption{Constant-round approximation distributed algorithms for \minDS on $H$-minor-free graphs, for various $H$. The bottom part of the table is about large $H$, on $t$ vertices where $t$ may be arbitrarily large.
}
\label{table:results}
\end{center}
\end{table*}
\myfootnotetext{foot}{1}{If there are at least three vertices, take all vertices with degree at least two, cf. \cite{LPW13,BCGW21}. This requires two rounds from the model, because the vertices do not know their degree and need one round to count their neighbors (by counting the number of received messages).}%
\myfootnotetext{foot}{2}{Take all the vertices. Such graphs have degree at most $t-1$, thus this is a $0$-round $t$-approximation since every dominating set has size at least $n/(\Delta+1)$ where $\Delta$ is the maximum degree of the graph. }

Our approach toward \Cref{TH:MAIN} is to design an algorithm as simple as possible and push all the complexity to its analysis, following a long tradition~\cite{Wawrzyniak14,BCGW21}. The intuition here is that a highly-connected $K_{2,t}$-minor-free graph has bounded radius, so the difficulty lies in handling small cuts, especially since we cannot decide locally if a vertex belongs to a small cut. We treat all vertices that are in a cut of size~$1$ or~$2$ in their small-distance neighborhood as we would vertices that are in a cut of size~$1$ or~$2$ in the whole graph (take all vertices in a cut of size~$1$, take all vertices in a cut of size~$2$ except those which are clearly a bad idea), then argue what remains\footnote{While $3$-connected $K_{2,t}$-minor-free graphs may have unbounded radius, such graphs admit many local $1$- or $2$-cuts, see \Cref{lem:cc_bounded_radius}.} is a number of connected components of bounded weak radius, which we can thus solve optimally by brute-force. Though the latter part comes with its own interesting challenges which we thankfully mostly outsource\footnote{While the paper is only available as a preprint and has seemingly not gone through a reviewing process, it seems to be generally considered to be correct and has even been refined in a doctoral thesis~\cite{solava2019fine}. For our own peace of mind, we have triple-checked that the pieces we need from that paper really do hold.} to a paper of Ding~\cite{K2tcaract}, the major conceptual contribution is in the analysis of the first part. To argue that taking all vertices that are locally separating (and similarly for vertices in a local cut of size $2$) is not too costly, we need to give some global discharging argument. We were able to do this using recent results on the asymptotic dimension, a notion introduced by Gromov in 1993 in the context of geometric group theory~\cite{gromov1993geometric}. We believe this tool, which we detail in \Cref{sec:asymptotic}, will be of further interest to the community of distributed graph algorithms.

\paragraph{Sketch of the Main Algorithm (\Cref{TH:MAIN}).}

The algorithm computes an approximation of \minDS in a $K_{2,t}$-minor-free graph $G$. It has three main steps:
\begin{enumerate}
    \item Compute the set $X$ of all vertices in ``local'' $1$-cuts, and add them to the solution.
    \item Compute the set $I$ of all ``interesting'' vertices in ``local'' $2$-cuts, and add them to the solution.
    \item Compute an optimal dominating set of all other undominated vertices in each component of $G - (X \cup I)$ using a brute-force approach and add them to the solution.
\end{enumerate}
Intuitively, a ``local'' $k$-cut is a minimal set of vertices that locally (up to some bounded radius) looks like a standard $k$-cut. This radius is a function of the size of $H$ and of $k \in\set{1,2}$. And, a vertex $u$ in a ``local'' 2-cut $\set{u,v}$ is ``interesting'' if $v$ does not dominate all vertices, except for at most one component attached to $\set{u,v}$. This is a rough explanation, and all formal definitions of ``local'' $k$-cuts and ``interesting'' vertices can be found in~\Cref{sec:Prelim} and \Cref{sec:asymptotic}.

The main challenge is to accurately tune the above radii in Step~1 and~2 to show that the approximation ratio (namely \RATIO) does not depend on the size of $H$, but only on the asymptotic dimension of the class and its control function (see \Cref{sec:asymptotic}). In contrast, the round complexity essentially relies on the diameter of the components as defined in Step~3, which is a function of the radii defined above.

\section{Preliminaries}\label{sec:Prelim}

\paragraph{General Definitions.}
In a graph $G$, a set $S\subseteq V(G)$ is a dominating set if and only if every vertex of $G$ is either in $S$ or adjacent to a vertex in $S$.
We denote by $\MDS(G)$ the minimum size of such a set.
Given some $B\subseteq V(G)$, a set $S\subseteq V(G)$ is $B$-dominating if and only if every vertex of $B$ is either in $S$ or adjacent to a vertex in $S$. 
In particular, if $N[B]$ denotes the closed neighbourhood of $B$, we can assume that $S\subseteq N[B]$.
Similarly, we denote by $\MDS(G,B)$ the minimum size of such a set.

A graph without true twins is a graph such that no two distinct vertices $u$ and $v$ are true twins, i.e. are such that $N[u] = N[v]$.
The \emph{true-twin-less} graph associated to $G$ is a largest subgraph of $G$ with no true twins.
Notice that there is a unique such unlabelled graph $G^-$ and $G^-$ can be computed in a constant number of rounds in the \LOCAL model. Furthermore, $\MDS(G^-)=\MDS(G)$. 

The \emph{weak diameter} of a set $S\subset V(G)$ is the largest distance in $G$ between two vertices $u,v\in S$.

\paragraph{Local Connectivity.}

The aim of this definition is to study cuts that can be recognized using a \LOCAL algorithm. Recall that a $k$-cut of a graph $G$ is a minimal subset of $k$ vertices whose removal increases the number of connected components of $G$. E.g., $1$-cuts are a.k.a cut-vertices. On classes of bounded asymptotic dimension, the set of local $k$-cuts is well-behaved for $k \leq 2$ (see \Cref{sec:asymptotic}). 

Here is a formal definition of a local cut\footnote{We did not find this exact same notion anywhere else in the literature, though it was probably considered previously.}.
By $N^r[v]$, we denote the set of all vertices at distance at most $r$ of $v$ in $G$.
\begin{definition}[Local cut]
    A subset of vertices $C$ of a graph $G$ is a \emph{$r$-local $k$-cut} if all vertices of $C$ are pairwise at distance at most $r$ in $G$, and $C$ is a $k$-cut of $G\left[\bigcup_{v\in C}N^r[v]\right]$. 
\end{definition}
We say $G$ is $r$-locally $k$-connected if $G$ has no $r$-local $k$-cuts. If there are no $r$-local $k$-cuts, then there are no $r'$-local $k$-cuts for any $r'>r$. Note that a $k$-cut is a $|V(G)|$-local $k$-cut. Therefore a locally $k$-connected graph is $k$-connected, that is, local connectivity is a stronger notion than connectivity.
All cuts that will be considered from now on are minimal cuts, i.e., no proper subset of the cut is also a cut with the ``same'' connected components. Intuition is not always an ally when it comes to $r$-local $k$-cuts, however we use the notion in a fairly basic manner here.

\section{Asymptotic Dimension}\label{sec:asymptotic}

In this section, we focus on defining asymptotic dimension (a non-trivial task, as it happens) and explaining how to exploit it, in the hope that others may be able to exploit it in turn.

The first pitfall is that asymptotic dimension only makes sense when defined for a whole graph class and not for a single graph.

Given a graph $G$, we say $S \subseteq V(G)$ is \emph{$D$-bounded} if $G[S]$ has weak diameter at most $D$. An \emph{$r$-component} of $S$ is a maximal subset $S' \subseteq S$ such that for any two vertices $u,v \in S'$, there is a sequence of vertices $u_1=u, u_2, u_3, \ldots, u_p=v$ in $S'$ such that any two consecutive vertices are at distance at most $r$ of each other in $G$. Put differently, an $r$-component of $S$ is exactly a connected component of $G^r$. \newline

The asymptotic dimension of a graph class $\mathcal{G}$ is the least integer $d$ such that there exists a function $f$ that satisfies the following conditions:
For any graph $G \in \mathcal{G}$, for any $r > 0$: 
\begin{itemize}
    \item $G$ has a cover $V(G) = \bigcup_{i=0}^d B_i$;
    \item each $r$-component of $B_i$ is $f(r)$-bounded.
\end{itemize}

A function $f$ witnessing that $\mathcal{G}$ has asymptotic dimension at most $d$ is called the \emph{control function} of $\mathcal{G}$.

Note that any finite graph class has asymptotic dimension $0$, as one can always take $f$ to be constant, always set to the number of vertices in a largest graph of $\mathcal{G}$, then take $B_0$ to be the whole graph regardless of $r$.

Trees, and more generally graph classes of bounded treewidth (resp. layered treewidth) have asymptotic dimension~$1$ (resp.~$2$). Planar graphs, and more generally the classes of $H$-minor-free graphs (for any fixed $H$) have asymptotic dimension~$2$ as shown by \cite{BBEGLPS24}: the dependency in $H$ only shows in the control function $f$. Dense graph classes may also have small asymptotic dimension. For example, it is sufficient that there is a quasi-isometry into a class having small asymptotic dimension\footnote{It is conjectured in \cite{BBEGLPS24} that any graph class forbidding some graph as a \emph{fat} minor should also have asymptotic dimension at most $2$.}. The class of chordal graphs, being quasi-isometric to the class of trees, has asymptotic dimension $1$.

Asymptotic dimension is a large-scale generalisation of weak diameter network decomposition which has been studied in distributed computing; a more refined notion of asymptotic dimension is called Assouad-Nagata dimension and its algorithmic form is related to weak sparse partition schemes. The interested reader is referred to~\cite{BBEGLPS24} for further details.

\subsection{When Local Properties Can Replace Global Properties}

Let us give a first application of asymptotic dimension. We say that a graph class $\D$ is $r$-locally-$\cC$ if for every $v\in V(G)$, $G[N^r[v]]\in\cC$. We first prove that a dominating set for $r$-locally-$\cC$ classes can be approximated with good approximation ratio if there is an approximation algorithm on $\cC$ and if the graph class we are approximating on has bounded asymptotic dimension. For a given graph $G$ and a \LOCAL algorithm $\A$ that returns a subset of vertices, we define $\A(G)$ as the set returned by $\A$ when run on $G$.

\begin{proposition}\label{PROP:DS_APPROX_LOCAL_CLASSES}
    Let $\cC$ be a hereditary class of graphs and $k\geq 0$. Let $\A$ be a local algorithm with round complexity $r\geq 1$ and with the following property: for every $G\in\cC$ and $S\subseteq V(G)$, $|\A(G)\cap S|\leq \alpha \cdot \MDS(G,N^k[S])$. Let $\D$ be a graph class with asymptotic dimension $d$ with control function $f$, and that is $(f(2k+3)+k+r)$-locally-$\cC$.
    Then $\A$ is an $\alpha (d+1)$-approximation algorithm on $\D$.
\end{proposition}

The proof can be found in \Cref{app:asdim_application}. We unfortunately managed to simplify our algorithm so as not to use \Cref{PROP:DS_APPROX_LOCAL_CLASSES}, but we decided to include it anyways because it showcases the interest of asymptotic dimension and may be of future use.

\subsection{Bounding the Number of Local \texorpdfstring{$1$}{1}-cuts and \texorpdfstring{$2$}{2}-cuts}\label{sec:bounding}

We first bound the number of vertices in local $1$-cuts and the number of vertices in so-called interesting local $2$-cuts.
\begin{lemma}\label{lem:MDS_local_1_cuts}
    For any $d\in\N$, and any graph class $\cC$ of asymptotic dimension at most $d$, there exists $c_{\ref{lem:MDS_local_1_cuts}}(d)$ and $m_{\ref{lem:MDS_local_1_cuts}}(\cC)$ such that for all graphs $G\in\cC$, the number of $m_{\ref{lem:MDS_local_1_cuts}}(\cC)$-local $1$-cuts in $G$ is at most $c_{\ref{lem:MDS_local_1_cuts}}(d)\MDS(G)$.
\end{lemma}
The proof of this can be found in \Cref{subsec:local_1_cuts}.

To extend \Cref{lem:MDS_local_1_cuts} to $2$-cuts, we need some restriction on the $2$-cuts considered: for example, a large tree with a single vertex adjacent to all its vertices admits many $2$-cuts but has a dominating set of size $1$. This motivates the following definition, which we will only use for $r\geq 2$.

A vertex $v\in C$ is \emph{$r$-interesting} if there exists some $r$-local $2$-cut $c=\{u,v\}$ such that:
\begin{itemize}
    \item $N[v]\not\subseteq N[u]$ and
    \item at least two connected components of $G[N^r[c]]-c$ contain each a vertex non-adjacent to $u$.
\end{itemize}

We are now ready to state the corresponding lemma for $2$-cuts.

\begin{lemma}\label{LEM:MDS_INTERESTING_LOCAL_2_CUTS}
    For any $d\in\N$, and any graph class $\cC$ of asymptotic dimension at most $d$, there exists $c_{\ref{LEM:MDS_INTERESTING_LOCAL_2_CUTS}}(d)$ and $m_{\ref{LEM:MDS_INTERESTING_LOCAL_2_CUTS}}(\cC)$ such that for all graphs $G\in\cC$, the number of interesting vertices in $m_{\ref{LEM:MDS_INTERESTING_LOCAL_2_CUTS}}(\cC)$-local $2$-cuts of $G$ is at most $c_{\ref{LEM:MDS_INTERESTING_LOCAL_2_CUTS}}(d) \MDS(G)$.
\end{lemma}
The proof of this can be found in \Cref{subsec:local_2_cuts}.

\section{Constant Approximation for Minimum Dominating Set}

\paragraph{Intuition and Explanation.}

One can assume that the graph contains no true twins, just like in the algorithm of \Cref{TH:3ROUND}.
The main idea of our algorithm is to take all vertices in $1$-cuts and $2$-cuts, in order to reduce the problem to $3$-connected graphs, where we can solve \minDS in constant time $O(t)$.
However, it is not possible to do this in constant round-complexity in the \LOCAL model.
Therefore, instead of considering $k$-cuts we consider sets of $k$ vertices that resemble a $k$-cut locally.
With a little luck, those vertices are actually $1$-cuts, but not all local $1$-cuts are $1$-cuts.
Indeed, consider a very long cycle. All vertices are local $1$-cuts but none are global $1$-cuts.
However, we can show that, if the graph has bounded asymptotic dimension, there exists some constant $r$ (that does not depend on the graph) such the number of $r$-local $1$-cuts is bounded above by a function linear in $\MDS(G)$.
Therefore, our algorithm can take all local $1$-cuts in the returned set.
The case of local $2$-cuts is more complicated: there are graphs with $\omega(\MDS(G))$ many vertices in $2$-cuts.
Indeed, consider a clique $G$ of size $n$.
Take an arbitrary vertex of the clique $u$, and for all vertices $v\neq u$ of the clique, add a new vertex $x_{uv}$ attached to $\{u,v\}=N(x_{uv})$.
This creates a graph $G$ which can be dominated solely by the vertex $u$.
However, all vertices of the original clique are in some minimal $2$-cut, as $\{u,v\}$ separates $x_{uv}$ from the rest of the clique -- there is an unbounded number of vertices in minimal $2$-cuts.
This leads us to the definition of interesting vertices in $2$-cuts, which we mentioned in \Cref{sec:bounding} and recall here.
A vertex $v\in C$ is \emph{$r$-interesting} for some $r\geq 2$ if there exists some $r$-local $2$-cut $c=\{u,v\}$ such that:
\begin{itemize}
    \item $N[v]\not\subseteq N[u]$ and
    \item at least two connected components of $G[N^r[c]]-c$ contain each a vertex non-adjacent to $u$.
\end{itemize}
The first condition is intuitive: one better take $u$ instead of $v$ if $N[v]\subseteq N[u]$.
The rough idea behind the second condition is that it allows us to create a nice mapping from interesting vertices to a minimum dominating set.
In more detail, we give a tree-like structure to the vertices in $2$-cuts, and show the second condition gives us the existence of some vertex $d$ in a MDS satisfying the following property: $d$ is a successor of $u$ in the tree-like structure, and is at bounded distance from $u$ in the tree-like structure.
Now, every interesting vertex $u$ can charge this vertex $d$.
We then show that this $d$ does not receive too many charges, because of the tree-like structure of the interesting $2$-cuts, and this allows us to bound the number of interesting vertices.
Let us do a quick recap of what the algorithm has done until now: it has taken all local minimal $1$-cuts and all interesting vertices in local minimal $2$-cuts.
Let us consider an arbitrary local minimal $2$-cut $\{u,v\}$.
There are three cases.
First, if both $u$ and $v$ are interesting, the algorithm has taken both vertices in its return set.
All components attached to $\{u,v\}$ can now be solved independently.
Secondly, if $u$ is interesting and $v$ is not interesting, the algorithm has taken $u$ in its return set.
Either $N[v]\subseteq N[u]$ and all components attached to $\{u,v\}$ can now be solved independently, or $u$ dominates all but one component attached to $\{u,v\}$.
In any case, all components attached to $c$ can now be solved independently too, as only one undominated component exists.
Finally, if neither $u$ nor $v$ is interesting, then as the graph is without true twins, one of $u$ or $v$ dominates all but one component attached to $\{u,v\}$.
Therefore, all components attached to $\{u,v\}$ can now be solved independently too.
Now, one can prove that the undominated vertices form connected components of bounded diameter, and our algorithm can brute-force and $G$-dominate the rest of the vertices in constant time.

\paragraph{The Algorithm.}
Let $t\ge 2$ be an integer. The following algorithm computes an approximation of \minDS on the class $\cC_t$ of $K_{2,t}$-minor-free graphs. The algorithm is divided into four steps:
\begin{enumerate}
    \item Remove true twins from the graph.
    \item Compute the set $X_1$ of all vertices in minimal $m_{\ref{lem:MDS_local_1_cuts}}(\cC_t)$-local $1$-cuts, and add them to the returned Dominating Set.
    \item Compute the set $I$ of minimal $m_{\ref{LEM:MDS_INTERESTING_LOCAL_2_CUTS}}(\cC_t)$-interesting vertices in $m_{\ref{LEM:MDS_INTERESTING_LOCAL_2_CUTS}}(\cC_t)$-local $2$-cuts, and add them to the returned Dominating Set.
    \item Let $U$ be the set of already dominated vertices that have no undominated neighbors. Dominate all other undominated vertices, in every component of $G - (X \cup I \cup U)$, using a brute-force approach.
\end{enumerate}

A more formal description of the algorithm is given below:
\begin{algorithm}
\caption{Constant approximation for \minDS}\label{algo:MDS}
\begin{algorithmic}
\Require An integer $t$, and $G$ a $K_{2,t}$-minor-free graph
\Ensure $S$ is a dominating set of $G$ with $|S|=O(\MDS(G))$
\State $G \gets \text{true-twin-less graph associated to }G$
\State $S \gets \{v\in V(G)\mid \{v\} \text{ is a } m_{\ref{lem:MDS_local_1_cuts}}(\cC_t) \text{-local minimal } 1\text{-cut of } G\}$
\State $S \gets S \cup \{v \in C\mid v\text{ is a } m_{\ref{LEM:MDS_INTERESTING_LOCAL_2_CUTS}}(\cC_t)\text{-interesting vertex of a }\newline \hspace*{10em} m_{\ref{LEM:MDS_INTERESTING_LOCAL_2_CUTS}}(\cC_t)\text{-local minimal } 2\text{-cut of } G\}$
\State $S \gets S \cup (\text{brute-forced minimum set of $G$ that dominates } G-N[S])$
\end{algorithmic}
\end{algorithm}

We can now state the main theorem of this section.
\begin{theorem}\label{TH:MAIN}
    For every integer $t \ge 2$, \Cref{algo:MDS} is a $O_t(1)$-round $\RATIO$-approximate deterministic \LOCAL algorithm for \minDS on $K_{2,t}$-minor-free graphs.
\end{theorem}

\begin{proof}
    \Cref{algo:MDS} clearly outputs a dominating set of $G$. By \Cref{lem:MDS_local_1_cuts,LEM:MDS_INTERESTING_LOCAL_2_CUTS},  the approximation ratio of the algorithm is $c_{\ref{lem:MDS_local_1_cuts}}(1)+c_{\ref{LEM:MDS_INTERESTING_LOCAL_2_CUTS}}(1)+1=\RATIO$. 
\end{proof}

It remains to argue that the number of rounds is bounded. To do this, we need to argue that the brute-forcing performed takes constant time.
\begin{lemma}\label{lem:cc_bounded_radius}
    For every integer $t$, if $\cC_t$ is the class of $K_{2,t}$-minor-free graphs, there exists $m_{\ref{lem:cc_bounded_radius}}(t)$ such that for every $G\in\cC_t$, if $X$ is the set of vertices in $m_{\ref{lem:MDS_local_1_cuts}}(\cC_t)$-local $1$-cuts, $I$ is the set of $m_{\ref{LEM:MDS_INTERESTING_LOCAL_2_CUTS}}(\cC_t)$-locally interesting vertices of $G$, and $U = \{u\in N[I \cup X] \mid N[u]\subseteq N[I \cup X]\}$, then every connected component of $G\setminus (I \cup X \cup U)$ has diameter at most $m_{\ref{lem:cc_bounded_radius}}(t)$.
\end{lemma}
The proof of this can be found in \cref{sec:bruteforce}.

The asymptotic dimension of $K_{2,t}$-minor-free graphs is $1$ by \cite{BBEGLPS24} -- $K_{2,t}$ is planar so $K_{2,t}$-minor-free graphs have bounded treewidth by the grid minor theorem.
For $\cC_t$ the class of $K_{2,t}$-minor-free graphs, the running time is $\max\{m_{\ref{lem:MDS_local_1_cuts}}(\cC_t),m_{\ref{LEM:MDS_INTERESTING_LOCAL_2_CUTS}}(\cC_t),m_{\ref{lem:cc_bounded_radius}}(t)\} = 3\max\set{f(5)+2, f(11)+5} + g(t) + 3$, where $f$ is the control function of the class of $K_{2,t}$-minor-free graphs, and $g$ the linear function given in~\cite[Lemma~6.3]{K2tcaract}.
Choosing $f(r) = (5r+18)t$ suffices, see~\cite[Lemma 7.1]{BBEGLPS24}; and the running time of the algorithm is linear in $t$.

The observant reader may be struck by the fact that the roles of $t$ and of the asymptotic dimension seem disjoint. This can be highlighted with the following variant, which computes a $(c_{\ref{lem:MDS_local_1_cuts}}(d)+c_{\ref{LEM:MDS_INTERESTING_LOCAL_2_CUTS}}(d)+1)$-approximation of $\MDS$ in a class of asymptotic dimension $d$, given its control function, with running time that depends on $f$ and the largest $K_{2,t}$-minor of the input graph but does not require prior knowledge of it.

\begin{algorithm}
\caption{Constant approximation for \minDS in a bounded $\asdim$ class}\label{algo:MDS2}
\begin{algorithmic}
\Require An integer $d$, a control function $f$, and $G$ a graph in a class $\mathcal{G}$ of asymptotic dimension $d$ with control function $f$
\Ensure $S$ is a dominating set of $G$ with $|S|=O(\MDS(G))$
\State $G \gets \text{true-twin-less graph associated to }G$
\State $S \gets \{v\in V(G)\mid \{v\} \text{ is a } m_{\ref{lem:MDS_local_1_cuts}}(\mathcal{G}) \text{-local minimal } 1\text{-cut of } G\}$
\State $S \gets S \cup \{v \in C\mid v\text{ is a } m_{\ref{LEM:MDS_INTERESTING_LOCAL_2_CUTS}}(\mathcal{G})\text{-interesting vertex of a }\newline \hspace*{10em} m_{\ref{LEM:MDS_INTERESTING_LOCAL_2_CUTS}}(\mathcal{G})\text{-local minimal } 2\text{-cut of } G\}$
\State $S \gets S \cup (\text{brute-forced minimum set of $G$ that dominates } G-N[S])$
\end{algorithmic}
\end{algorithm}

We can now state the following stronger version of \Cref{TH:MAIN}.
\begin{theorem}\label{TH:MAIN2}
    For every integer $d$ and control function $f$, \Cref{algo:MDS2} is a $O_t(1)$-round $(c_{\ref{lem:MDS_local_1_cuts}}(1)+c_{\ref{LEM:MDS_INTERESTING_LOCAL_2_CUTS}}(1)+1)$-approximate deterministic \LOCAL algorithm for \minDS on graphs in a class of asymptotic dimension $d$ with control function $f$, where $t$ is the (unknown) size of a largest $K_{2,t}$-minor in the input graph.
\end{theorem}

As an added note, if one wishes to have an algorithm for \minVC instead of \minDS, it suffices to take all $m_{\ref{LEM:MDS_INTERESTING_LOCAL_2_CUTS}}(\cC_t)$-local $2$-cuts instead of just $m_{\ref{LEM:MDS_INTERESTING_LOCAL_2_CUTS}}(\cC_t)$-interesting vertices. On the analysis side, one can prove a simpler variant of \Cref{LEM:MDS_INTERESTING_LOCAL_2_CUTS} that bounds the number of vertices in local $2$-cuts with respect to the size of a minimum vertex cover. Therefore, both \Cref{TH:MAIN} and \Cref{TH:MAIN2} extend to the context of $\minVC$.

We conclude the ``non-technical'' part of the paper with the following result, which shows a different trade-off: linear approximation in constant number of rounds.

\begin{theorem}\label{TH:3ROUND}
    For every integer $t \ge 2$, there is a $3$-round $(2t-1)$-approximate deterministic \LOCAL algorithm (resp. $t$-approximate) for \minDS (resp. \minVC) on $K_{2,t}$-minor-free graphs.
\end{theorem}

As outerplanar graphs are a subfamily of $K_{2,3}$-minor-free graphs, this result generalizes the $5$-approximation algorithm of~\cite{BCGW21} on outerplanar graphs in the case of the LOCAL model. The proof of \Cref{TH:3ROUND} can be found in \cref{app:linear_algo}.

\let\O\oldO

\section{Proofs}

\subsection{Proof of Proposition \ref{PROP:DS_APPROX_LOCAL_CLASSES}: From Local to Global}
\label{app:asdim_application}

Let $G\in\D$. By the definition of the asymptotic dimension applied to $G$, there is a cover $B_0, B_1, \dots, B_{d}$ of $G$ where each $2k+3$-component of a $B_i$ is $f(2k+3)$-bounded. Note that each $B_i$ contains distinct $2k+3$-components which are of distance at least $2k+4$ from each other. With an abuse of notation, when we write $B \in B_i$ we mean that $B$ is a $2k+3$-component of $B_i$. That is, we treat $B_i$ as the set of its $2k+3$-components. Let $\A$ be an $\alpha$-approximation algorithm for $\cC$ with round complexity $r$.
Let us run the same algorithm on graphs from $\D$.
By the covering property, $|\A(G)| \leq \sum_{i=0}^{d} |\A(G)\cap B_i|$.

Let $i\in\{0,1,\dots,d\}$ and $B\in B_i$.
We have the following:
\begin{claim}
    $G[N^{k+1}[B]]\in \cC$.Ptetre que si tu passes 
\end{claim}
Let $v\in B$.
Because $\D$ is $(f(2k+3)+k+r)$-locally $\cC$, $G'=G[N^{f(2k+3)+k+r}[v]]\in\cC$.
Moreover, $N^{k+1}[B]\subseteq G'$ because $B$ has weak diameter at most $f(2k+3)$ and because $r\geq 1$. 
Therefore, as $\cC$ is hereditary, $G[N^{k+1}[B]]\in\cC$.

By assumption on $\A$, $|\A(G')\cap B| \leq \alpha\cdot\MDS(G',N^k[B]) \leq \alpha\cdot\MDS(G,N^k[B])$ as $G'$ contains $N^{k+1}_G[B]$.
As $N^r_G[B]\subseteq V(G')$, vertices in $B$ have the same distance-$r$ neighborhood in $G$ and $G'$.
Therefore, $|\A(G)\cap B| = |\A(G')\cap B| \leq \alpha\cdot\MDS(G,N^k[B])$.
Using this, along with \Cref{lem:MDS_union_bound}
and that every $B_i$ is partitioned into $2k+3$-components, we get for every $i\in\{0,1,\dots,d\}$ that $|\A(G)\cap B_i| \leq \sum_{B\in B_i} \alpha \cdot \MDS(G,N^k[B]) \leq \alpha\cdot \MDS(G)$.
Putting everything together we get 
\[
    |\A(G)|
    ~\leq~ \sum_{i=0}^{d} |\A(G)\cap B_i|
    ~\leq~ \alpha \cdot (d+1) \cdot \MDS(G) ~.
\]

This completes the proof of \Cref{PROP:DS_APPROX_LOCAL_CLASSES}.

\subsection{Proof of Lemma~\ref{lem:MDS_local_1_cuts}: Bounding the Number of Vertices in Local \texorpdfstring{$1$}{1}-cuts}
\label{subsec:local_1_cuts}

We will need the following lemma for the rest of the proofs in this section.
\begin{lemma}\label{lem:MDS_union_bound}
    Let $G$ be a graph and let $R_0, R_1, \dots, R_k\subseteq V(G)$ subsets of vertices such that all $N[R_i]$ are pairwise disjoint. Then
    \[
        \sum_{i=0}^k \MDS(G,R_i) ~\leq~ \MDS(G) ~.
    \]
\end{lemma}

Let us now prove \Cref{lem:MDS_local_1_cuts}.

We did not try to optimize the constants $c_{\ref{lem:MDS_local_1_cuts}}(d)$ and $m_{\ref{lem:MDS_local_1_cuts}}(\cC)$. 
Let $f$ be the $d$-dimensional control function of the graph class. 
We prove the lemma for $c_{\ref{lem:MDS_local_1_cuts}}(d)=3\cdot(d+1)$ and $m_{\ref{lem:MDS_local_1_cuts}}(\cC)=f(5)+2$.
Without loss of generality, we can assume $G$ is connected.
We will first prove there are not too many $1$-cuts in some $S\subseteq V(G)$ compared to $\MDS(N[S])$. 
\begin{claim}\label{claim:1_cuts_MDS_bound}
    Let $G$ be a graph and $S\subseteq V(G)$.
    Let $C$ be the set of minimal $1$-cuts of $G$.
    Then $|C\cap S|~\leq~3\cdot\MDS(G,N[S])$.
\end{claim}
Let us prove this claim.
Without loss of generality, we can assume $G$ is connected.
Let $C$ be the set of $1$-cuts of $G$ and $B$ the set of maximal $2$-connected components of $G$.
Let $T$ be the bipartite graph with vertex set $B\cup C$ and with edge set $E(T)=\{(b,c)\in B\times C \mid c\in b \}$.
$T$ is sometimes called the block-cut tree of $G$ and can be proven to be a tree. Moreover, note that all leaves of $T$ are in $B$.
Let $S\subseteq V(G)$, and $D\subseteq N^2[S]$ a dominating set of $N[S]$.
We prove that $|C\cap S|\leq 3|D|$.
If $C\cap S=\emptyset$, we are done. Otherwise, root $T$ at an arbitrary cut-vertex $r$.
We have the following:
\begin{claim}
    Let $c\in C\cap S$. Then there exists $b$ such that $c \in b \in B$ such that $b\cap D\neq \emptyset$.
\end{claim}
This is because all vertices of $C\cap S$ must be dominated by some vertex of $D$. Therefore, either $c\in D$ and then we are done as there exists some $b\in B$ such that $c\in b$, or either there exists $a\in D\cap N[c]$. This $a$ must be contained in some neighboring $2$-connected component, therefore there exists $b\in B$ such that $a\in b$, i.e. $b\cap D \neq \emptyset$.

In the following, we create a mapping from vertices of $C\setminus D$ to $D$.
Consider $c\in C\setminus D$.
There are 3 different cases.
\begin{itemize}
    \item Either $c$ has a child $b\in B$ with some $d\in b\cap D\cap N(b)$, and in this case, we map $c$ to $d$.
    \item Or either $c$ has an descendant $c'\in C\cap N(c)$ (at distance $2$ in $T$), and the previous claim still applies to $c'$: $c'$ has an descendant $b'\in B$ in $T$ such that there exists $d\in d\cap D\cap N(c')$.
    We map $c$ to $d$.
    \item Or $c$ has no descendant in $C\cap \cap N(c)$.
    In this case, as $c$ cannot be a leaf of $T$, there exists a child $b\in B$ of $c$, with the property that $b\setminus C \neq \emptyset$.
    As $D$ dominates $N[S]$, $b$ must contain a vertex $d\neq c$ of $D$. We map $c$ to $b$.
\end{itemize}
Therefore, we created a mapping from vertices of $C\setminus D$ to vertices of $D$.
Furthermore, each vertex from $D$ can appear at most twice in a preimage.
Indeed, for some $d\in D$, only an ancestor in $C$ at distance $1$ or $3$ in $T$ can be mapped to it.
In conclusion, $|C|\leq |C\cap D|+|C\setminus D| \leq |D|+2|D|=3|D|$.

We can now get a bound on the number of local $1$-cuts.
\begin{claim}
    Let $G$ be a graph of asymptotic dimension $d$ with control function $f$.
    Then, the number of $(f(5)+2)$-local $1$-cuts is bounded by $3(d+1)\cdot\MDS(G)$.    
\end{claim}
This will directly imply the statement of the lemma.
Let $r$ a positive integer and $C$ be the set of $(f(5)+2)$-local cuts of $G$.
Fix $B\subseteq V(G)$ such that $G[B]$ has weak diameter $f(5)$.
We claim the following:
\begin{claim}\label{claim:MVC_1_cuts_local_to_global}
    Every $(f(5)+2)$-local $1$-cut of $G$ that is in $B$ is also a $1$-cut of $G[N^2[B]]$.
\end{claim}

Indeed, let $v\in B$ be a $(f(5)+2)$-local $1$-cut of $G$ and $a,b\in N(v)$ separated by $v$, i.e. every $ab$-path of $G$ of $N^{f(5)+2}[v]$ contains $v$.
Notice that $N^2[B] \subseteq N^{f(5)+2}[v]$ because $B$ has weak diameter at most $f(5)$.
$a$ and $b$ are separated by $v$ in $G[N^2[B]]$, because $a,b\in N^2[B]$ and because no $ab$-path in $N[B]\setminus \{v\}$ exists.
Therefore, \Cref{claim:MVC_1_cuts_local_to_global} is proven.

Using this fact and with the help of \Cref{claim:1_cuts_MDS_bound}, we can bound the number of $1$-cuts of $G[N^2[B]]$ in $B$ by $3\cdot \MDS(G[N^2[B]],N[B])\leq 3\cdot \MDS(G,N[B])$.
By the definition of the asymptotic dimension, there is a cover $B_0, B_1, \dots, B_d$ of $G$ where the $5$-components of a $B_i$ are $f(5)$-bounded.
Note that each $B_i$ contains distinct $5$-components which are of distance at least $6$ from each other. With an abuse of notation, when we write $B \in B_i$ we mean that $B$ is a $5$-component of $B_i$. That is, we treat $B_i$ as the set of its $5$-components.
As $B_0, B_1, \dots, B_d$ is a cover of $V(G)$ by subsets of diameter at most $f(5)$, we can bound the number of $(f(5)+2)$-local $1$-cuts of $G$ by summing the number of $1$-cuts of all the $G[N^2[B_i]]$'s.
We get
\[
    |C| ~\leq~ \sum_{i=0}^{d} \sum_{B\in B_i} 3\cdot\MDS(G,N[B]) ~.
\]
Notice that because the $B_i$'s are partitioned into their $5$-components, all elements of $\{N^2[B]\mid B\text{ connected component of }B_i\}$ are pairwise disjoint. 

Therefore by \Cref{lem:MDS_union_bound}, we get
\[
    |C| ~\leq~ \sum_{i=0}^d 3\cdot\MDS(G) = 3(d+1)\cdot\MDS(G) ~.
\]
This finishes the proof of \Cref{lem:MDS_local_1_cuts}.

\subsection{Proof of Lemma~\ref{LEM:MDS_INTERESTING_LOCAL_2_CUTS}: Bounding the Number of Interesting Vertices}
\label{subsec:local_2_cuts}

When discussing global $2$-cuts and not local ones, we say $v$ is interesting if there exists a $2$-cut $c=\{u,v\}$ such that:
\begin{itemize}
    \item $N[v]\not\subseteq N[u]$ and
    \item at least two connected components of $G-c$ contain each a vertex non-adjacent to $u$.
\end{itemize}
Moreover, $v$ is called a \emph{friend} of $u$, and a $2$-cut $\{u,v\}$ where $u$ is interesting and $v$ is a friend of $u$ is called \emph{interesting}. 
If $u$ only has the second property, it is called \emph{almost-interesting}.

Two $2$-cuts $c_1,c_2$ of $G$ are said to be \emph{crossing} the two following conditions are verified:
\begin{itemize}
    \item the two vertices of $c_1$ are in different components of $G-c_2$, and
    \item the two vertices of $c_2$ are in different components of $G-c_1$. 
\end{itemize}

Before bounding the number of interesting vertices in local $2$-cuts, we first need to arrange the interesting cuts in a tree-like fashion, i.e. we want to build a bounded number of families of $2$-cuts such that each member of the family contains interesting $2$-cuts that are all pairwise non-crossing, and such that each interesting vertex appears in one family member, along with one of its friends.

One can easily see that a family of size  $2$ does not suffice by considering $C_6$.
If we want to only take interesting cuts in this graph, we need to take the $3$ opposing cuts.
In more detail, if the vertices $\{a,b,c,d,e,f\}$ of $C_6$ appear in clockwise order $a$,$b$,$c$,$d$,$e$ and $f$, then we need to take the interesting cuts $\{a,d\}$, $\{b,e\}$ and $\{c,f\}$.

To create our new $2$-cut forest for interesting vertices, we need to introduce SPQR trees.

\paragraph{SPQR Trees.}
An SPQR tree is a tree data structure that represents the decomposition of a $2$-connected graph into its $3$-connected components.
The construction of an SPQR tree can be accomplished in linear time and SPQR are known to have applications in dynamic graph algorithms and graph drawing.

An SPQR tree $T$ is an unrooted tree where each node $\mu$ corresponds to an undirected skeleton graph $G_\mu$ that can be one of the following four types.

\begin{itemize}
    \item \bm{$S$}\textbf{-node:} $G_\mu$ is a cycle containing three or more vertices. This represents series composition in series-parallel graphs.
    \item \bm{$P$}\textbf{-node:} $G_\mu$ corresponds to a dipole graph, a multigraph with two vertices and three or more edges, analogous to parallel composition.
    \item \bm{$Q$}\textbf{-node:} $G_\mu$ corresponds to a dipole connected by two parallel edges: one real and one virtual. This serves as a trivial case for graphs with two parallel edges. We will not consider these types of nodes.
    \item \bm{$R$}\textbf{-node:} $G_\mu$ is a $3$-connected graph that is neither a cycle nor a dipole.
\end{itemize}

Edges $xy$ between nodes in the SPQR tree are associated with two directed virtual edges, one from $G_x$ and the other from $G_y$. Each edge in $G_x$ can be a virtual edge for at most one edge in the SPQR tree.

The SPQR tree represents a $2$-connected graph $G_T$, constructed as follows.
If $xy\in E(T)$ is associated with the virtual edge $ab\in E(G_x)$, and with the virtual edge $cd\in E(G_y)$, then identify $a$ with $c$ and $b$ with $d$, and delete the two virtual edges.
Notably, no two adjacent $S$ or $P$ nodes are allowed, ensuring the uniqueness of the SPQR tree representation for a graph $G$.
When such conditions are met, the graphs $G_x$ associated with the nodes of the SPQR tree are the triconnected components of $G$.

\begin{proposition}[folklore]\label{prop:2_cuts_in_SQPR_tree}
    Let $T$ be a SPQR tree of a graph $G$ (without $Q$ nodes) and let $\{u,v\}$ be a $2$-cut of $G$.
    Then one of the following holds:
    \begin{itemize}
        \item $u,v$ are two endpoints of a virtual edge of a $R$-node, or
        \item $u,v$ are the two vertices of a $P$-node that has at least two virtual edges, or
        \item $u,v$ are two endpoints of a virtual edge of a $C$-node, or
        \item $u,v$ are two non-adjacent vertices of a $C$-node.
    \end{itemize}
\end{proposition}

We can now build our tree-like structure for interesting vertices.

\paragraph{Interesting $2$-cuts Forests.}
An \emph{interesting $2$-cut forest} $F=(T_1,T_2,T_3)$ of $G$ consists of three trees $T_1,T_2$ and $T_3$ whose vertices contain subsets of $V(G)$.
Like SPQR trees, $T_i$ contains nodes that are induced subgraphs with some virtual edges added.
Every $T_i$ has nodes can be of three types: $A$-nodes, for $1$-cuts, $C$-nodes, for interesting $2$-cuts, and $R$-nodes for the rest.
$T_2$ has still $C$-nodes and $R$-nodes.

If $G$ is has no $1$-cut or interesting $2$-cut, $T_i$ consists of a $R$ single node $\mu=G$.

If $G$ has a $1$-cut $v$, we construct $T_i$ inductively.
First, we add a $A$-node $\mu$ to $T_i$.
The graph associated to $\mu$ consists of the vertex $v$.
Secondly, let $C_1,C_2, \dots, C_k$ be the connected components of $G-v$. 
Let $G_j$ be the graph $G[C_j \cup \{v\}]$.
Build a corresponding $2$-cut tree $T_{G_j}$ for the graph $G_j$.
Let $\mu_j$ be the (unique) node in $T_{G_j}$ that contains $v$.
We can now construct $T_i$ by taking the union of all $T_{G_j}$'s and connecting all $\mu_j$'s to $\mu$.

Now, let us handle the case of interesting $2$-cuts.
We do this by going through a SPQR tree $T$ of $G$ and building sets of $2$-cuts $P_1, P_2$ and $P_3$ with the following properties:
\begin{itemize}
    \item for every globally almost-interesting vertex $u$ of $G$, there exist some $i$ and some friend $v$ of $u$ such that $\{u,v\}\in P_i$, and
    \item for every $i$, the $2$-cuts in $P_i$ are pairwise non-crossing.
\end{itemize}
The second property allows us to transform $P_i$ into a $2$-cut tree $T_i$, as follows.
First, take some arbitrary $c\in P_i$ and a $C$-node $\mu$ to $T_i$.
The graph associated to $\mu$ consists of vertices of $c=\{u,v\}$ and a real (respectively virtual) edge $uv$ if $uv$ is a real (resp. virtual) edge
of G.
Secondly, let $C_1,C_2, \dots, C_k$ be the connected components of $G-c$. 
Let $G_j$ be the graph $G[C_j \cup c]$ to which we add a real edge
$uv$ if $uv\notin G[C_j \cup c]$.
Build a corresponding interesting $2$-cut tree $T_{G_j}$ for the graph $G_j$.
Let $\mu_j$ be the (unique) node in $T_{G_j}$ where $uv$ is real.
We can now construct $T_i$ by taking the union of all $T_{G_j}$'s, making $uv$ virtual in all $\mu_j$'s, and connecting all $\mu_j$'s to $\mu$.

We first build the sets $P_1, P_2$ and $P_3$.
We then prove the two wanted properties in \Cref{prop:interesting_2_cut_tree}.

Add the vertices $u$ and $v$ to $P_1$ if:
\begin{itemize}
    \item $u,v$ are two endpoints of a virtual edge of a $R$-node of $T$, or if
    \item $u,v$ are the two vertices of a $P$-node of $T$ that has at least two virtual edges.
\end{itemize}
Let us now handle the case of $C$-nodes.
Let $\mu$ be a $C$ node of the $T$.
If $\mu$ contains more than $6$ nodes.
Let $v_0, v_1, \dots, v_{k-1}$ be the nodes of $\mu$ in the order of the cycle.
First, put all $\{u,v\}$ in $P_1$ if $uv$ is a virtual edge.
Secondly, we add some $2$-cuts to the $P_i$'s depending on the values of $k$:
\begin{enumerate}
    \item If $k\geq 8$ and $k$ is even: add to $P_1$ the $2$-cuts $\{v_0,v_{k-3}\}$, $\{v_1,v_{k-4}\}$, \dots, and $\{v_{(k/2)-3},v_{k/2}\}$, and to $P_2$ the $2$-cuts $\{v_{(k/2)-2},v_{k-1}\}$ and $\{v_{(k/2)-1},v_{k-2}\}$.
    \item If $k\geq 8$ and $k$ is odd: add to $P_1$ the $2$-cuts $\{v_0,v_{k-3}\}$, $\{v_1,v_{k-4}\}$, \dots, $\{v_{((k-1)/2)-3},v_{(k+1)/2}\}$ and $\{v_{((k-1)/2)-3},v_{(k-1)/2}\}$. Add to $P_2$ the $2$-cuts $\{v_{((k-1)/2)-2},v_{k-1}\}$ and $\{v_{((k-1)/2)-1},v_{k-2}\}$.
    \item If $k=7$, add to $P_1$ the $2$-cut $\{v_0,v_3\}$ and $\{v_0,v_4\}$, to $P_2$ the $2$-cut $\{v_1,v_5\}$ and to $P_3$ the $2$-cut $\{v_2,v_6\}$.
    \item If $k=6$, add to $P_1$ the $2$-cut $\{v_0,v_3\}$, to $P_2$ the $2$-cut $\{v_1,v_4\}$ and to $P_3$ the $2$-cut $\{v_2,v_5\}$.
    \item If $k\leq 5$ but $G\neq C_k$, suppose without loss of generality that the edge $v_0 v_1$ is virtual. Moreover, suppose that it is the only virtual edge of the $C$-node.
    If $k=5$, add to $P_1$ the $2$-cut $\{v_0,v_2\}$ and to $P_2$ the $2$-cut $\{v_1,v_4\}$.
    \item If $k\leq 5$ but $G\neq C_k$ and the edges $v_0 v_1$ and $v_0 v_{k-1}$ are virtual: add to $P_1$ all the $2$-cuts $\{v_0,v_i\}$ for $i=2,3,\dots, k-2$. Moreover, if $k=5$, add to $P_2$ the $2$-cut $\{v_1,v_{k-1}\}$.
    \item If $k\leq 5$ but $G\neq C_k$ and there exists $i\in\{2,3,\dots,k-2\}$ such that the edges $v_0 v_1$ and $v_i v_{i+1}$ are virtual:
    add to $P_1$ all the $2$-cuts $\{v_0,v_j\}$ for $j=2,3,\dots, i$, and add to $P_2$ all the $2$-cuts $\{v_1,v_j\}$ for $j=i+1,i+2,\dots, k-1$.
\end{enumerate}

\begin{proposition}\label{prop:interesting_2_cut_tree}
    Let $P_1,P_2$ and $P_3$ be built as described above. 
    Then the two following properties are verified:
    \begin{itemize}
        \item for every globally interesting vertex $u$ of $G$, there exist some $i$ and some friend $v$ of $u$ such that $\{u,v\}\in P_i$, and
        \item for every $i$, the $2$-cuts in $P_i$ are pairwise non-crossing.
    \end{itemize}
\end{proposition}
\begin{proof}
    One can easily see that $2$-cuts of $P_i$ in the same $C$-node are taken in such a way that they do not cross.
    This also applies $2$-cuts inside the same nodes inside the same $P$ or $R$-node.
    Moreover, $2$-cuts inside the different nodes of the SPQR tree cannot cross by \Cref{prop:2_cuts_in_SQPR_tree}.
    Therefore, the second property is proven.

    We will show that every globally interesting vertex $u$ of $G$ appears in some $P_i$ with one of its friends.
    As we take all $2$-cuts in $R$-nodes and $P$-nodes, then by \Cref{prop:2_cuts_in_SQPR_tree} the only case where we could have not taken an interesting vertex with its friend is inside a $C$-node.
    We prove that is however not the case.
    Let us consider a $C$-node $\mu$ with $k$ vertices.
    We go through all the possible cases. 

    \begin{itemize}
        \item If $k\geq 8$ and $k$ is even: one can verify that all $2$-cuts chosen are interesting, and that every vertex of the cycle is in one of the chosen cycles. Therefore we are done with this case.
        \item The same applies if $k\geq 8$ and $k$ is odd, or if $k=7,6$, by checking each case.
        \item If $G=C_k$ with $k\leq 5$, there are no interesting vertices.
        \item If $k\leq 5$ but $G\neq C_k$, and the only virtual edge of $\mu$ is $v_0 v_1$.
        Consider the node connected to $\mu$ by this virtual edge.
        If it is not a $P$-node, name this node $\mu'$.
        If it is a $P$-node, let $\mu'$ be one of its neighbor different from $\mu$ (it exists as every $P$-node has degree at least $2$ in $T$).
        Let us handle the case $k=5$ first.
        Notice that by definition, $V(\mu')\setminus V(\mu) \neq \emptyset$.
        Let us prove that $2$-cuts $\{v_0,v_2\}$ and $\{v_1,v_4\}$ are interesting. 
        First, their one component in $\mu$ is not fully dominated by either of the vertices of the cuts.
        Moreover, in each of the two $2$-cuts, one of the vertices is not connected to any vertex of $V(\mu')\setminus V(\mu)$.
        Finally, $N[v_1]\not\subseteq N[v_4]$, $N[v_0]\not\subseteq N[v_2]$, and inversely.
        Therefore the $2$-cuts $\{v_0,v_2\}$ and $\{v_1,v_4\}$ are indeed interesting. 
        The only vertex that needs checking is now $v_3$.
        We claim that $v_3$ is not interesting.
        Indeed, any $2$-cut in $\mu$ containing $v_3$ is not interesting.
        Moreover, $v_3$ cannot have a friend outside of $V(\mu)$ by \Cref{prop:2_cuts_in_SQPR_tree}.
        Therefore, $v_3$ is not interesting.
        All interesting vertices are now taken in some $P_i$, and all chosen $2$-cuts in $\mu$ are chosen.
        We are done with this case.
        If $k=4$ (resp. $k=3$) then for similar reasons, vertices $v_2$ and $v_3$ (resp. $v_2$) are not interesting.
        In both cases, all possibly interesting vertices are taken in the $2$-cut $\{v_0, v_1\}$ (taken because the edge $v_0 v_1$ is virtual).
        All other $2$-cuts in $\mu$ containing either $v_0$ or $v_1$ are not interesting.
        Therefore, we are done with this case.
        \item If $k\leq 5$ but $G\neq C_k$, and the edges $v_0 v_1$ and $v_0 v_{k-1}$ are virtual.
        Consider the node connected to $\mu$ by the virtual edge $v_0 v_1$.
        If it is not a $P$-node, name this node $\mu'_1$.
        If it is a $P$-node, let $\mu'_1$ be one of its neighbor different from $\mu$ (it exists as every $P$-node has degree at least $2$ in $T$).
        Define similarly the node $\mu'_2$ for the virtual edge $v_0 v_{k-1}$.
        Let $v'_i\in V(\mu'_i)\setminus V(\mu)$.
        For $i=2,3,\dots, k-2$, the $2$-cuts $\{v_0,v_i\}$ are interesting because $v_i$ is not adjacent to $v'_1$ nor $v'_2$, and $N[v_0]\not\subseteq N[v_i]$, and inversely.
        If $k=5$, the $2$-cut $\{v_1,v_{k-1}\}$ is interesting because $v_1 v'_2\notin E(G)$, $v_1 v_3\notin E(G)$ and $N[v_{k-1}]\not\subseteq N[v_1]$, and $N[v_1]\not\subseteq N[v_{k-1}]$.
        If $k=5$, all vertices are taken in interesting $2$-cuts.
        If $k=4$, $v_0$ and $v_2$ are taken in interesting $2$-cuts.
        All other $2$-cuts contained only in $\mu$ containing either $v_0$ or $v_1$ are not interesting.
        The $2$-cuts $\{v_0,v_1\}$ and $\{v_0,v_3\}$ are taken anyway, so it they are interesting, we took them in $P_1$. 
        Otherwise, it does not matter: if there is an interesting $2$-cut containing $v_1$ or $v_3$, it will be taken in another node.
        The case $k=3$ is similar.
        \item If $k\leq 5$ but $G\neq C_k$ and there exists $i\in\{2,3,\dots,k-2\}$ such that the edges $v_0 v_1$ and $v_i v_{i+1}$ are virtual.
        Without loss of generality, we can consider that $i=2$ and $k\in\{4,5\}$.
        Consider the node connected to $\mu$ by the virtual edge $v_0 v_1$.
        If it is not a $P$-node, name this node $\mu'_1$.
        If it is a $P$-node, let $\mu'_1$ be one of its neighbor different from $\mu$ (it exists as every $P$-node has degree at least $2$ in $T$).
        Define similarly the node $\mu'_2$ for the virtual edge $v_i v_{i+1}$.
        Let $v'_i\in V(\mu'_i)\setminus V(\mu)$.
        If $k=5$, for similar reasons as in the last case, because of the existence of $v'_1$ and $v'_2$, the cuts $\{v_0,v_2\}$, $\{v_1,v_3\}$ and $\{v_1,v_4\}$ are interesting. 
        Therefore if $k=5$, all vertices are taken in interesting $2$-cuts and we are done.
        If $k=4$, the $2$-cut $\{v_0,v_2\}$ (resp. $\{v_1, v_3\})$ is interesting if and only if $N[v_0]\not\subseteq N[v_2]$ or $N[v_2]\not\subseteq N[v_0]$ (resp. $N[v_1]\not\subseteq N[v_3]$ or $N[v_3]\not\subseteq N[v_1]$).
        If they are interesting, we took them, otherwise it does not matter: we have taken all possibly interesting $2$-cuts.
        We are done with this case.
    \end{itemize}
    Therefore, the first property is proved.
\end{proof}

We say that $T_i$ \emph{displays} the vertices $u$ and $v$ through the node $\mu$.
A vertex that is part of an interesting minimal $2$-cut but that is not displayed by $T$ is called \emph{hidden}.
$F$ \emph{displays} $u$ if at least one of the $T_i$'s displays $u$.

\begin{corollary}\label{cor:interesting_cut_tree_to_forest}
    Let $G$ be a $2$-connected graph, $S\subseteq V(G)$ and $k$ be a constant depending on the graph.
    Suppose that for any interesting $2$-cut tree, the number of vertices $u\in S$ that appear with some friend $v$ such that $\{u,v\}$ is an interesting cut displayed by $T$ is bounded by $k$.
    Then if $C$ is the set of interesting vertices in $2$-cuts, $|C\cap S|\leq 3k$.
\end{corollary}

\paragraph{Bounding the Number of Interesting Vertices.}
We now can prove \Cref{LEM:MDS_INTERESTING_LOCAL_2_CUTS}.
We did not try to optimize the constants $c_{\ref{LEM:MDS_INTERESTING_LOCAL_2_CUTS}}(d)$ and $m_{\ref{LEM:MDS_INTERESTING_LOCAL_2_CUTS}}(\cC)$. 
Let $f$ be the $d$-dimensional control function of the graph class. 
We prove the lemma for $c_{\ref{LEM:MDS_INTERESTING_LOCAL_2_CUTS}}(d)=22\cdot(d+1)$ and $m_{\ref{LEM:MDS_INTERESTING_LOCAL_2_CUTS}}(d)=f(11)+4$.

Without loss of generality, one can assume that the graph is $2$-connected. Indeed, if it is not, one can split $G$ into $2$-connected component and do the analysis on those components.
Let $T$ be a $2$-cut tree of $G$, rooted at an arbitrary $C$ node.
Let $D$ be a dominating set of $N^4[S]$ in $G$ using vertices in $N^5[S]$.
Let $I$ be the set of interesting vertices displayed in $T$ and $I'=I\cap S$.
Let us prove this claim first.
\begin{claim}\label{claim:interesting_tree_bound}
    $|I'| \leq 6\cdot\MDS(G,N^4[S])$.
\end{claim}
Let $u\notin D$ be an interesting vertex displayed in $T$ and $v$ be a friend of $u$, i.e. a vertex such that $c=\{u,v\}$ is a $2$-cut with two components of $G-c$ not dominated entirely by $v$, and with $N[u]\not\subseteq N[v]$.
Let us first prove the following claim:
\begin{claim}\label{claim:interesting_distance_5}
    There exists $d\in D$ such that $d_G(u,d)\leq 5$ and $d$ is lower in $T$ than $u$. Moreover, the interesting-ness of a vertex is certified by vertices at distance at most $4$.
\end{claim} 
Let $\mu$ a node of $T$ below $c$ that is contained in a component $C$ of $G-c$ not fully dominated by $v$.
Let $w\in C$ a vertex that is not dominated by $v$ and that minimizes $d_G(u,w)$.
$w$ will be our witness of interesting-ness of one of the components.
If one wishes to get the witness of another component $C'$, one can apply the same technique on $w'\in C$, a vertex not dominated by $v$ that minimizes $d_G(u,w')$.  
Note that $w$ is well-defined by definition of $C$.
Let us first prove that $d_G(u,w)\leq 4$.
Note that $w\in C$ and therefore $w$ is lower in $T$ than $u$.
If there exists some $x\in (N(u)\setminus N(v))\cap C$, one can take $w=x$ and then $d_G(u,w)=1$.
Otherwise, take $y\in N(w)$ such that $d_G(u,y)<d_G(u,w)$.
By minimality of $d(u,w)$, $y\in N(v)$.
Take $x\in (N(u)\cap N(v))\cap C$.
Such a $x$ always exists, because $c$ is a minimal $2$-cut (i.e. $N(u)\cap C\neq \emptyset$).
The path $wyvxu$ exists, therefore $d(u,w)\leq 4$.
One can chose a dominating vertex of $D$ adjacent to $w$ if $w\notin D$,
or $d=w$ if $w\in D$.
By the triangle inequality, $d(u,v)\leq 5$, and \Cref{claim:interesting_distance_5} is proven.

Let $q:I'\setminus D\to D$ a function that we define later as our charging function. 
Let $d\in D$ be a vertex of a node lower than $u$ in $T$ and such that $d(u,d)\leq 5$, chosen to be one of the highest in $T$ among all possible candidates.
Note that $d$ is well-defined by \Cref{claim:interesting_distance_5}.
We set $q(u):=d$ and say that $u$ \emph{charges} $d$.
Now, we bound the size of the preimages of $q$.

For a fixed $d\in q(I'\cap S)$, choose $u\in q^{-1}(\{d\})$ highest in $T$.
Again, by \Cref{prop:interesting_2_cut_tree} there exists $v$ a friend of $u$ displayed in the same node as $u$.
$v$ can be chosen highest-in-$T$ among all the possible candidates.
Let $\mu_u:=\{u,v\}$ and $\mu_d$ be the highest-in-$T$ $R$-node that contains $d$. 
Notice that even though $d$ may appear lower in $T$ than $\mu_d$, there cannot be any interesting vertex charging $d$ lower than $\mu_d$, as this would mean some interesting vertex charges a vertex higher than itself, which is not possible.
Let $\mu'$ the lowest-in-$T$ $C$-node that is higher than $\mu_d$.
Let $F$ be the set of $C$-nodes of $T$ that are between $\mu_u$ and $\mu'$ (both non-included), that form interesting $2$-cuts, and such that one of the interesting vertices in the cut is in $S\setminus \{u\}$.
There are two possible cases.
\begin{itemize}
    \item Either there exists some $c\in F$ such that $c\cap N[v] = c\cap \mu' = \emptyset$.
    In this case, let $w$ be the interesting vertex of $c$. 
    We get a contradiction because $w$ cannot be dominated by $v$ nor a vertex below it, as $d$ is lower than $\mu'$. 
    \item Therefore, all $c\in F$ either contain some $w\in \mu'$ or contain a vertex dominated by $v$.
    Notice that if $w$ exists, it is unique for all $c\in F$.
    Therefore, $F$ is partitioned in two sets: $F_2$, the set of $c\in F$ that contain $w$, and $F_1$, the set of $c\in F\setminus F_2$ that contain some vertex dominated by $v$.
    Furthermore, all the $c\in F_1$ appear higher in the tree than the $c\in F_2$.
    We claim that $q^{-1}(\{d\})\subseteq \mu_u\cup \mu'\cup \mu_c$ for some $\mu_c\in F\cup \{\emptyset\}$.
    Let us first prove that then all $c\in F_1$ but one $2$-cut contain $v$.
    Let $\mu_c$ be the highest $C$-node of $F_1$ that does not contain $v$, if it exists.
    If it does not exist, set $\mu_c = \emptyset$.
    Suppose there exists a $c'\in F_1$ strictly below $\mu_c$.
    All vertices of $c'$ cannot be dominated by $v$.
    We are not in the first case, therefore $c'$ must contain some vertex of $\mu'$, i.e. $c\in F_2$.
    Let $c\in F_1\setminus \{\mu_c\}$ if $\mu_c$ exists, else let $c\in F_1$.
    Let $x\in q^{-1}(\{d\})$ be an interesting vertex in $c$ different from $u$.
    $x$ and its neighbors can only be dominated by $v$, therefore $N[x]\subseteq N[v]$ and $x$ cannot be interesting, i.e. $x$ does not exist.
    Similarly, let $c\in F_2$ and let $x\in q^{-1}(\{d\})$ be an interesting vertex in $c$ that is not in $\mu_c\cup \mu'$, if it exists.
    $x$ and its neighbors can only be dominated by $w$, therefore $N[x]\subseteq N[w]$ and $x$ cannot be interesting, i.e. $x$ does not exist.
\end{itemize}

We therefore get that $|q^{-1}(\{d\})|\leq 6\cdot \MDS(G,N^4[S])$.
This proves \Cref{claim:interesting_tree_bound}.
Moreover, by \Cref{cor:interesting_cut_tree_to_forest}, we get the following claim.
\begin{claim}\label{claim:interesting_relative_bound}
    The number of interesting vertices of $G$ in $S$ is at most $19\cdot \MDS(G,N^4[S])$.
\end{claim}

We can now get a bound on the number of interesting vertices in local $2$-cuts.
\begin{claim}
    Let $G$ be a graph of asymptotic dimension $d$ with control function $f$.
    Then, the number of $(f(11)+5)$-local interesting vertices is bounded by $22(d+1)\cdot\MDS(G)$.    
\end{claim}
This will directly imply the statement of the lemma.
Fix $B\subseteq V(G)$ such that $G[B]$ has diameter $f(11)$.
We claim the following:
\begin{claim}\label{claim:local_2_cut_implies_2_cut}
    Every vertex in a $(f(11)+5)$-local $2$-cut of $G$ that is in $B$ is also a in a $1$-cut or a $2$-cut of $G[N^5[B]]$.
\end{claim} 
Indeed, let $v\in B$ such that $c=\{u,v\}$ is a $(f(11)+5)$-local $2$-cut of $G$.
Let $a,b\in N(v)\setminus c$ two distinct vertices separated by $v$, i.e. there are not two internally-disjoint $ab$-paths in $G[N^{f(11)+5}[c]]-c$.
Notice that $N[B] \subseteq N^{f(11)+5}[c]$ because $B$ has weak diameter at most $f(11)$.
$ab$ are separated by $c$ in $G[N[B]]$, because $a,b\in N[B]$ and because no two disjoint $ab$-paths in $N[B]\subseteq c$ exist.
Therefore, \Cref{claim:local_2_cut_implies_2_cut} is proven.

Every interesting vertex of $B$ is either a $1$-cut of $G[N^5[B]]$ or in a $2$-cut of $G[N^5[B]]$ and interesting in $G[N^5[B]]$.
Indeed, by \Cref{claim:interesting_distance_5}, the interesting-ness of a vertex is certified by a vertex at distance at most $4$.

By \Cref{claim:1_cuts_MDS_bound}, we can bound the number of vertices in $1$-cuts of $G[N^5[B]]$ in $B$ by $3\MVC(G[N^5[B]],N[B])\leq 3\MVC(G,N[B])$, and by \Cref{claim:interesting_relative_bound}, we can bound the number of interesting vertices of $2$-cuts of $G[N^5[B]]$ in $B$ by $19\MVC(G[N^5[B]],N^4[B])\leq 19\MVC(G,N^4[B])$.

By the definition of the asymptotic dimension, there is a cover $B_0, B_1, \dots, B_d$ of $G$ where the $11$-components of $B_i$ are $f(11)$-bounded. 
Note that each $B_i$ contains distinct $11$-components which are of distance at least $12$ from each other. With an abuse of notation, when we write $B \in B_i$ we mean that $B$ is a $11$-component of $B_i$. That is, we treat $B_i$ as the set of its $11$-components.
As $B_0, B_1, \dots, B_d$ is a cover of $V(G)$ by subsets of diameter at most $f(11)$, we can bound the number of $(f(11)+5)$-local $2$-cuts of $G$ by summing the number of $2$-cuts of all the $G[N^4[B_i]]$'s.
Let $I$ be the set of interesting vertices in $(f(11)+5)$-local $2$-cuts of $G$.
We get
\[
    |I| ~\leq~ \sum_{i=0}^d \sum_{B\in B_i} (3\cdot\MDS(G,N[B]) + 19\cdot\MDS(G,N^4[B])) ~.
\]
Notice that because the $B_i$'s are partitioned into their $11$-components, all elements of $\{N^5[B]\mid B\text{ connected component of }B_i\}$ are pairwise disjoint. 
Therefore by \Cref{lem:MDS_union_bound}, we get
\[
    |I| ~\leq~ \sum_{i=0}^d 22\cdot\MDS(G) = 22(d+1)\cdot\MDS(G) ~.
\]
This finishes the proof of \Cref{LEM:MDS_INTERESTING_LOCAL_2_CUTS}.

\subsection{Proof of Lemma~\ref{lem:cc_bounded_radius}: Brute-forcing is fast enough}\label{sec:bruteforce}

To prove this, we need a few results from~\cite{K2tcaract}.
\paragraph{The structure of $3$-connected $K_{2,t}$-minor-free graphs.}\label{par:ding}
Let us give a quick overview of their result.
The author defines two types of graphs.
The first type is a generalization of outerplanar graphs.
Let $G$ be a graph with a specified Hamiltonian cycle $C$, which is called the reference cycle.
Edges of $G$ that are not part of $C$ are called chords.
Two non-incident chords $ab$ and $cd$ are said to cross if the vertices $a, c, b, d$ appear in that order around $C$.
$G$ is said to be of type-I if each chord crosses at most one other chord. Additionally, if two chords $ab$ and $cd$ do cross, then either both $ac$ and $bd$ are edges of $C$, or both $ad$ and $bc$ are edges of $C$.
The class of all type-I graphs is denoted by $\mathcal{P}$.

Let $H$ be a type-I graph with reference cycle $C$, and let $ab$ and $cd$ be two distinct edges of $C$.
Assume all chords of $C$ lie between the two paths of $C \setminus \{ab, cd\}$.
If $ab$ and $cd$ share an endpoint, say $a = d$, then $H$ is called a fan with corners $a, b, c$.
The vertex $a$ is called the center of the fan, and the number of chords is called the fan's length.
If $ab$ and $cd$ have no shared endpoints, then for any subset $F \subseteq \{ ab, cd \}$, $H \setminus F$ is called a strip with corners $a, b, c, d$, provided the minimum degree of $H \setminus F$ is at least two.
The radius of a strip $H$ as $\max\{d_H(h,y)\mid h\in H, x\in\{a,b,c,d\}\}$.
Strips are proven to be $K_{2,5}$-minor-free in~\cite{K2tcaract}, and it is not hard to see that if the radius of a strip is large, then its corners form local cuts.

Let $G$ be a graph.
Adding a fan or a strip to $G$ means identifying the corners of a fan or a strip (which is disjoint from $G$) with distinct vertices of $G$.
An augmentation of $G$ is obtained by adding disjoint fans and strips to $G$, with the condition that if two corners are identified with the same vertex of $G$, then one of them is the center of a fan, and the other is either a center of a fan or a corner of a strip.
For any integer $m$, $\B_m$ is defined as the class of graphs on at most $m$ vertices, and $\A_m$ is the class of all augmentations of graphs in $\B_m$. 
\begin{proposition}[Corollary 1.6 of~\cite{K2tcaract}]\label{prop:ding}
    Let $t$ be an integer and $\cC_t$ be the class of $K_{2,t}$-minor-free graphs.
    Then there exists some $m_{\ref{prop:ding}}(t)$ such that $\cC_t\subseteq \A_{m(t)}$.
\end{proposition}

\paragraph{Embedding components into a $3$-connected graph.}
In the following, we prove \Cref{lem:cc_bounded_radius}.

Already, note that our algorithm takes all $1$-cuts. We can therefore assume that the rest of the graph that remains to be solved is $2$-connected: if $C_1$ is the set of all $1$-cuts of $G$, let $V_1, V_2, \dots, V_k$ be the different connected components of $G-C_1$.
Consider for every $i\in\{1,2,\dots, k\}$, $G_i := G[V_i\cup N[V_i]\cap C_1]$ and do the analysis on this graph.
If there exists a constant approximation algorithm for $2$-connected graphs, then there exists a constant approximation algorithm for $1$-connected graphs.
We want to do a similar reasoning for $2$-cuts. However, our algorithm does not take all $2$-cuts, but only interesting vertices of these $2$-cuts.

In $G$, let $X$ be the set of $m_{\ref{lem:cc_bounded_radius}}(t)$-local $1$-cuts, $X_2$ be the set of vertices inside an $m_{\ref{LEM:MDS_INTERESTING_LOCAL_2_CUTS}}(\cC_t)$-local $2$-cuts, $I$ the set of $m_{\ref{LEM:MDS_INTERESTING_LOCAL_2_CUTS}}(\cC_t)$-important $C_1$ the set of $1$-cuts, and $C_2$ the set of $2$-cuts.
Finally, let $Y=X\cup I$ and $U$ be the set of vertices in $\{u\in N[Y] \mid N[u]\subseteq N[Y]\}$ the set of vertices dominated by $Y$ that do not have neighbors not dominated by Y.

Let $C$ be a connected component of $G-(Y \cup U)$.
We want to embed $G[C]$ into a $3$-connected graph $G'$ without creating any $K_{2,t}$ minor in $G'$.  
Let $G_1 = G[C\cup (N[C]\cap Y)]$.
We add the edges $\{u,v\}$ that are formed by contracting the connected components of $G-G_1$.
We claim that $G_1$ is $3$-connected.
Let $c$ be a $1$-cut or $2$-cut of $G_1$, separating it into at least two connected components $A$ and $B$.
It cannot be that $c\subseteq Y\cup U$, as $C$ is connected and would not intersect $c$.
Therefore, $c$ is not a $1$-cut of $G$, nor an interesting cut. If it is not a $2$-cut, there exists an $AB$-path in $G$ that does not intersect $c$, and this path would still exist after the contraction of the definition of $G_1$, contradiction.
If $c$ is a $2$-cut of $G$, one of the vertices $v$ of $c=\set{u,v}$ is not interesting.
There are two possibilities: either $N[v]\subseteq N[u]$, or all connected components of $G-c$ but one (assuming $G$ is connected) are dominated entirely by $u$.
In the first scenario, there are two subcases.
If $u\in Y$, then $v\in U$ and $c\subseteq Y\cup U$, contradiction.
Otherwise, $u\notin Y$. Note that $N[u]\not\subseteq N[v]$ as $u$ and $v$ would be true twins.
$v$ has to dominate all but one component of $G-c$.
Then this also applies for $u$, and we enter the second scenario.
In the second scenario, delete all vertices that are in connected components of $G-c$ that are dominated entirely by $u$, and add the edge $uv$ if it does not exist.
Again, this creates a graph with diameter at least one less, and this does not create any new $2$-cuts.
The resulting graph $G'$ is $3$-connected, as all $2$-cuts from $G'$ have been handled.

\paragraph{Bounding the diameter of components.}
By \Cref{prop:ding}, there exists some $m_{\ref{prop:ding}}(t)$ such that $G'$ is an augmentation of some graph $\overline{G'}$ of bounded size, to which strips and fans are added. 
We would like to say that $G'$ has bounded radius.
$\overline{G'}$ and the attached fans have bounded radius, and we only have to bound the size of the strips of $G'$.
Note that if some strip has radius at least $3 m_{\ref{LEM:MDS_INTERESTING_LOCAL_2_CUTS}}(\cC_t)$, its corners in $\overline{G'}$ form two $m_{\ref{LEM:MDS_INTERESTING_LOCAL_2_CUTS}}(\cC_t)$-local $2$-cuts $c_1$ and $c_2$ in $G'$.
Those cuts are also $m_{\ref{LEM:MDS_INTERESTING_LOCAL_2_CUTS}}(\cC_t)$-local $2$-cuts in $G$, because $G'$ is an induced minor of $G$: suppose that $c$ is a local cut in $G'$ but not in $G$.
There would be a short path in $G$ between two different connected components of $G'-c$. However, as $G'$ is an induced minor of $G$, the path still exists and can only be shorter. 
Furthermore, $c_1$ and $c_2$ are both in $Y\cup U$. 
Suppose they are not. Let $c_1=\set{u,v}$. As the strip has a long radius, it cannot be that $N[u]\subseteq N[v]$ or that $N[v]\subseteq N[u]$.
Without loss of generality, $u$ dominates all but one connected component of $G-c$.
This means that $u$ is connected to some vertex on the other end of its own strip, which is impossible.
Finally, $G'$ cannot include strips of radius at least $3m_{\ref{LEM:MDS_INTERESTING_LOCAL_2_CUTS}}(\cC_tt)$.
Therefore, $G'$ has radius bounded by $3m_{\ref{LEM:MDS_INTERESTING_LOCAL_2_CUTS}}(\cC_tt) + m_{\ref{prop:ding}}(t)$, $C$ also has radius bounded by $m_{\ref{lem:cc_bounded_radius}}(t) = 3m_{\ref{LEM:MDS_INTERESTING_LOCAL_2_CUTS}}(\cC_tt) + m_{\ref{prop:ding}}(t) + 3$. 


\subsection{Proof of Theorem \ref{TH:3ROUND}: A linear approximation in constant rounds}
\label{app:linear_algo}

The goal of this sub-section is to build a simple $(2t-1)$-approximation algorithm that has constant round complexity (that does not depend on $t$) in $K_{2,t}$-minor-free graphs.

We need the following result from Ore.
\begin{lemma}[\cite{O62}]\label{lem:Ore}
    Every $n$-vertex graph without isolated vertices (a vertex with no adjacent edges) has a dominating set of size at most $\frac{n}{2}$.
\end{lemma}

When some graph $G$ is considered, $\gamma(v)$ will denote the minimum number of vertices different from $v$ needed to dominate $N[v]$.
We will also denote the set of vertices whose neighborhood cannot be dominated by less than two vertices (other than the original vertex) by $D_2(G)=\{v \in V(G)\ \mid \gamma(v) \geq 2\}$.

\begin{lemma}\label{lem:H_existence}
Let $G$ a graph, $S\subseteq V$ be a set of vertices, and $D\subseteq N^2[S]$ a minimum dominating set of $N[S]$ in $G$ of size $k$. Then there exists some minor $H$ of $G[N^2[S]]$ such the following hold:
\begin{enumerate}
    \item $V(H)=A\sqcup B$ and $|B|=k$,
    \item $H[A]$ is edgeless and $d_H(a)\geq 2$ for all $a\in A$ ($d_H(v)$ is the degree of $v$ in $H$),
    \item and $|A| \geq \frac{1}{2}|(D_2\cap S)\setminus D|$.
\end{enumerate}
\end{lemma}
\begin{proof}
    Let $G' = G[N^2[S]]$. 
    Let us start by constructing a minor as in Lemma 1 of \cite{KSV21}.
    Take $D$ a dominating set of size $k$, and write $D=\{d_1, d_2, \dots, d_k\}$.
    Define $H$ as the minor of $G'$ by contracting the branch sets $b_i$ defined as follows. $$b_i = N_{G'}[d_i]\setminus (D_2\setminus D \cup \bigcup_{j<i} N[d_i] \cup \{d_{i+1}, \dots, d_k\})$$
    Let $B= \{b_1, b_2, \dots, b_k\}$ and $A = (D_2 \cap S)\setminus D$.
    For any $v\in (D_2\cap S) \setminus D$, we have $d_H(v) \geq 2$, as $D$ is a dominating set of $N[S]$.
    Here we need an additional trick: we remove edges $uv$ in triangles of the form $u,v,d$ where $u,v \in (D_2\cap S)\setminus D$ and $d\in D$.
    This is always possible while keeping $d_H(v) \geq 2$ for any $v\in (D_2 \cap S)\setminus D$.
    Indeed, $D$ is a dominating set of $G$ and $\gamma(v)\geq 2$. Since $\gamma(v)\geq 2$, there exist two paths in $G$ of length at most $2$ from $v$ to distinct vertices of $D$.
    Moreover, these two paths only intersect at $v$.
    One can check that these paths also exist in $H$.
    We now want to contract some edges so that every non-contracted vertex left in $(D_2\cap S)\setminus D$ is adjacent to two vertices of $B$ in $H$.
    Consider the induced subgraph $H[(D_2\cap S)\setminus D]$.
    Let $I$ be its set of isolated vertices and $J=(D_2\cap S)\setminus(D\cup I)$ the rest of the vertices.
    Every vertex of $I$ is adjacent to two vertices of $B$ in $H$.
    However, it is not the case for vertices in $J$.
    By definition, $H[J]$ is a graph with no isolated vertices.
    Let $D'$ be a minimum dominating set of $H[J]$, which is of size at most $\frac{1}{2} \cdot |J|$ by \Cref{lem:Ore}.
    Every vertex $j\in J$ is adjacent to some $b_{k_j}$ (because $D$ is a dominating set).
    Contract the edges $jb_{k_j}$ for $j\in D'$, i.e. for every $j\in D'$, add the vertex $j$ to the branch set $b_{k_j}$ and remove $j$ from $A$.
    Because $D'$ is a dominating set of $J$, every vertex in $J\setminus D'$ is now adjacent to two vertices of $B$ in $H$.
    Moreover, because $|D'|\leq \frac12 \cdot |J| \leq \frac12 \cdot |(D_2\cap S)\setminus D|$, we have $|A| \geq \frac12 \cdot |(D_2\cap S)\setminus D|$.
    Delete all the edges from $H[A]$. This ends the proof.
\end{proof}

Using the same notation as in the above lemma, the next lemma bounds the size of $A$ with respect to $|B|$ on $K_{2,t}$-minor-free graphs.

\begin{lemma}\label{lem:H_size}
    Let $G=(A\sqcup B,E)$ a $K_{2,t}$-minor-free graph such that $G[A]$ is edgeless and every vertex $v\in A$ has degree at least $2$. Then $|A| \leq (t-1)|B|$.
\end{lemma}
\begin{proof}
    First, apply some preprocessing to the graph:
    if there exists some $a\in A$ and $b,b'\in N(a)$ such that $b$ and $b'$ are not connected in $G[B]$, then we can contract the edge $ab$, and from now on consider the newly contracted graph instead of $G$.
    This creates a new ``red'' edge $bb'$. Put it in $R$, the set of red edges.
    Thinking about red edges as paths of length $2$ allows us to do a fine-grain analysis of the size of $|A|$.
    Keep repeating this step whenever some such $a$ exists, and let $G'=(A'\sqcup B,E')$ be the final graph obtained after this procedure.
    \begin{figure}[H]
        \centering
        \begin{tikzpicture}[font=\boldmath,draw=primary,every node/.style=vertex, every path/.style={very thick}]
            \draw[rounded corners,draw=none,fill=secondary] (-2, -1) rectangle (3, 0.5) {};
            \node[fill=none] at (-1,-0.2) {$G[B]$};
            \node[label=above:{$a$}] (1) at (1,2) {};
            \node[label=below:{$b$}] (2) at (0,0) {};
            \node[label=below:{$b'$}] (3) at (2,0) {};
            \node[label=below:{$b$}] (4) at (6,1) {};
            \node[label=below:{$b'$}] (5) at (8,1) {};
        	\draw [-{Stealth[length=3mm, width=2mm]}] (3,1) -- (5,1);
            \draw (1) -- (2);
            \draw (1) -- (3);
            \draw[densely dashed,decorate,decoration=snake] (2) -- (3);
            \draw[tertiary] (4) -- (5);
        \end{tikzpicture}
        \caption{An explanation of the preprocessing procedure. Here the dashed squiggly edge represents the lack of path in $G[B]$.}
    \end{figure}
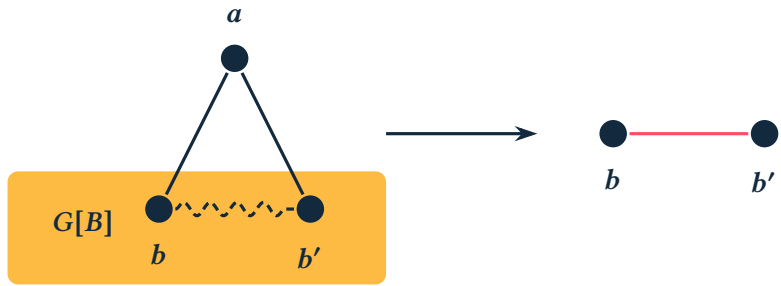

    We prove by induction on $|B|$ that $|A'| + |R| \leq (t-1)|B|$.

    The base case with $B=\emptyset$ is trivial.
    If $B$ is not empty, let $C$ be a connected component of $G'[B]$.
    We can find some $c\in C$ such that $G'[C\setminus\{c\}]$ is connected (for instance, one can take a leaf in a BFS tree of $G'[C\setminus\{c\}]$).
    Let $R(c) := \{cb \in E' | b \in V(B)\}\cap R$ denote the set of red edges that touch $c$, i.e. the red edges that are not in $G'[B\setminus\{c\}]$.
    Let $a \in N(c)\cap A'$ and $b \in N(a)\setminus\{c\}$. Then $b\in C$, i.e. $b$ is in the same connected component of $G'[B]$ as $c$.
    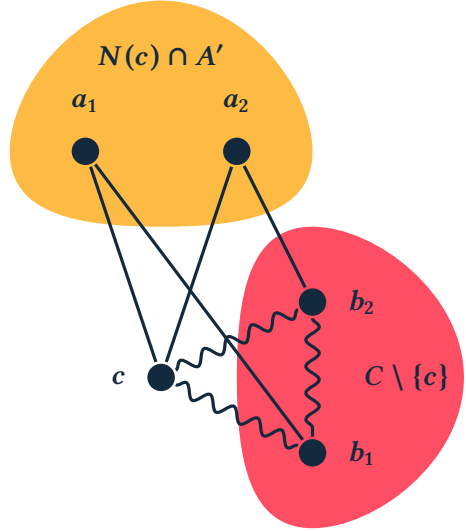
\begin{figure}[H]
    \centering
    \begin{tikzpicture}[font=\boldmath,draw=primary,every node/.style=vertex,every path/.style={very thick}]
		\draw[fill=secondary,draw=none] plot [smooth cycle, tension=1] coordinates {(-2, 3) (0, 5) (2, 3) (0, 2)};
		\draw[fill=tertiary,draw=none] plot [smooth cycle, tension=1] coordinates {(1, 0) (2, 2) (4, 0) (2, -2)};
		\node [label=left:{$c$}] (0) at (0, 0) {};
		\node [label=right:{$b_1$}] (1) at (2, -1) {};
		\node [label=above:{$a_1$}] (2) at (-1, 3) {};
		\node [label=above:{$a_2$}] (3) at (1, 3) {};
		\node [label=right:{$b_2$}] (4) at (2, 1) {};
		\node [fill = none] at (0, 4.25) {$N(c)\cap A'$};
		\node [fill = none] at (3.25, 0) {$C\setminus \{c\}$};

		\draw (2) to (0);
		\draw (0) to (3);
		\draw [decorate,decoration=snake] (0) to (4);
		\draw [decorate,decoration=snake] (4) to (1);
		\draw [decorate,decoration=snake] (1) to (0);
		\draw (2) to (1);
		\draw (3) to (4);

    \end{tikzpicture}
    \caption{The vertices at the top are in $A'$ and the vertices on bottom are in $B$. In the yellow region, the neighbors $a_i$ of $c$ in $A'$, which are all connected to a vertex $b_i \in C\setminus\{c\}$ in the red region. $c$ and the $b_i$ are all in the same connected component of $G[B]$.}
    \end{figure}

    Thus, we have the bound $|N(c)\cap A'| + |R(c)|\leq t-1$, as otherwise we would obtain a $K_{2,t}$ minor by contracting the vertex set $C\setminus\{c\} \neq \emptyset$.
    Remove the vertex $c$ from $B$ and the vertices $N(c)\cap A'$ from $A'$, to maintain the condition on the vertex degrees of $A'$.
    Notice that for each $a\in A'\setminus N(c)$, if $b,b'\in N(a)$ then $b$ and $b'$ are still connected in $G'[B]$. Indeed, if they were connected before but not anymore, we would have $b,b'\in C$, but as $C\setminus\{c\}$ is connected, we would get a contradiction.
    This new smaller graph contains all the red edges except the ones in $R(c)$.
    We have removed at most $t-1$ vertices from $A'$ or edges from $R$, so by applying induction hypothesis we have that $|A'| + |R| \leq (t-1)|B|$.
    Finally,  $|A| = |A'| + |R| \leq (t-1)|B|$.
\end{proof}

Moreover, we prove that in a graph with no true twins, $D_2$ is a dominating set.
\begin{lemma}
    Let $G$ be a graph with no true twins and let $v\in V(G)$ such that $\gamma(v)=1$. Then there exists some $u\in V(G)$ such that $\gamma(u)\geq 2$ and $N[v]\subseteq N[u]$.
\end{lemma}
\begin{proof}
    Let $v \in V(G)$ such that $\gamma(v) = 1$.
    Therefore, there exists some $u$ such that $N[v]\subsetneq N[u]$ because $G$ contains no true twins.
    Take $u$ such that $N[u]$ is maximal.
    Then $\gamma(u)\geq 2$. Indeed, if we had $\gamma(u) = 1$, there would be some other vertex $u'$ such that $N[u]\subsetneq N[u']$, because $G$ does not contain true twins.
    Contradiction.
\end{proof}
    
Using the preceding lemmas, we get the following.

\begin{corollary}\label{cor:D2_size}
    Let $G$ be a graph, and $S\subseteq V(G)$ a subset of vertices such that $G[N^2[S]]$ is $K_{2,t}$-minor free.
    Then $|D_2(G)\cap S| \leq (2t-1)\MDS(N[S])$.
\end{corollary}
\begin{proof}
    Let $H$ be the minor defined in \Cref{lem:H_existence}.
    Apply \Cref{lem:H_size} on $H$. We get that $|D_2(G)\setminus D| \leq  2|A| \leq 2(t-1)\MDS(G)$, so $|D_2(G)| \leq \MDS(G) + |D_2(G)\setminus D| \leq (2t-1)\MDS(G)$.
\end{proof}

We can now prove that the following is an $(2t-1)$-approximation algorithm in $K_{2,t}$-minor-free graphs.
\begin{enumerate}
    \item Make the graph without true twins. $u$ and $v$ are true twins in $G$ if $N[u]=N[v]$. The true-twin-less graph associated to $G$ is largest subgraph of $G$ with no true twins.
    \item Return the set $D_2 = \{ v\in V(G) \mid \not\exists u\in V(G-v), N[v] \subseteq N[u]\}$. It is the set of all vertices whose neighborhood cannot be dominated by a single vertex other than itself.
\end{enumerate}
The round complexity of this algorithm is constant, and does not depend on $t$.
Moreover, from \Cref{cor:D2_size} and the fact that the true-twin-less graph associated to $G$ has the same domination number as $G$, we get that this is indeed a $(2t-1)$-approximation algorithm.
This proves \Cref{TH:3ROUND}.

\begin{acks}
\begin{itemize}
    \item
    The first three authors have been partially supported by the French ANR projects ENEDISC (ANR-24-CE48-7768) and TEMPOGRAL (ANR-22-CE48-0001).  The fourth author was supported by the Deutsche Forschungsgemeinschaft (DFG, German Research Foundation) under Germany's Excellence Strategy – The Berlin Mathematics Research Center MATH+ (EXC-2046/1, project ID:390685689).
    \item
    The authors would like to thank Linda Cook and Sergey Norin for inspiring discussions, and the reviewers for helpful comments.
\end{itemize}
\end{acks}

\balance

\providecommand{\etalchar}[1]{$^{#1}$}

\balance

\end{document}